%% file: main.tex
\documentclass[]{article}

\usepackage[a4paper,total={6in,8in}]{geometry}
\usepackage[square,numbers]{natbib}
\usepackage{amsmath}
\usepackage{amssymb}
\usepackage{amsthm}
\usepackage{algorithm}
\usepackage{algorithmic}
\usepackage{graphicx}
\usepackage{float}
\usepackage{caption}
\usepackage{subcaption}
\usepackage[hidelinks]{hyperref}
\usepackage{authblk}
\usepackage{booktabs}

\graphicspath{{figures/}}

\newtheorem{definition}{Definition}
\newtheorem{theorem}{Theorem}
\newtheorem{lemma}{Lemma}
\newtheorem{corollary}{Corollary}

\newtheorem{proposition}{Proposition}

\newcommand{\PF}{\mathrm{PF}}
\newcommand{\IM}{\mathrm{IM}}
\newcommand{\UT}{\mathrm{UT}}
\newcommand{\NSW}{\mathrm{NSW}}
\newcommand{\GGI}{\mathrm{GGI}}
\newcommand{\EG}{\mathrm{EG}}
\newcommand{\NWU}{\mathrm{NW}_{U}}
\newcommand{\NWbeta}{\mathrm{NW}_{\beta}}

\author[1]{Joshua Caiata}
\author[2]{Carter Blair}
\author[1]{Kate Larson}
\affil[1]{University of Waterloo}
\affil[2]{Harvard University}
\affil[ ]{\texttt{jcaiata@uwaterloo.ca, carterblair@g.harvard.edu, kate.larson@uwaterloo.ca}}

\title{Procedural Fairness in Multi-Agent Bandits}
\date{}

\begin{document}

\maketitle

\begin{abstract}
In the context of multi-agent multi-armed bandits (MA-MAB), fairness is often reduced to outcomes: maximizing welfare, reducing inequality, or balancing utilities. However, evidence in psychology, economics, and Rawlsian theory suggests that fairness is also about process and who gets a say in the decisions being made. We introduce procedural fairness as equal voice and formalize it for MA-MABs as equal representation within a policy over arms, using a core-stable Nash welfare objective based on representation rather than utility. Empirical results confirm that fairness notions based on optimizing for outcomes sacrifice equal voice and representation, while the sacrifice in outcome-based objectives (like equality and utilitarianism) is minimal under procedurally fair policies. We further prove that different fairness notions prioritize fundamentally different and incompatible values, highlighting that fairness requires explicit normative choices. This paper argues that procedural legitimacy deserves greater focus as a fairness objective and provides a framework for putting procedural fairness into practice.
\end{abstract}

\input{sections/introduction}
\input{sections/multi-agent-bandits}
\input{sections/procedural-fairness}

\input{sections/learning-procedural-fairness}

\input{sections/comparison-framework}
\input{sections/theoretical-results}
\input{sections/experiments}
\input{sections/discussion-conclusions}

\bibliographystyle{plainnat}
\bibliography{references}

\end{document}

%% file: sections/introduction.tex
\section{Introduction}

From the Magna Carta to the words that open constitutions and charters around the world, we have long understood that dignity and fairness require more than simply providing a good outcome. To be fair requires treating each person as an equal, entitled to a voice in the decisions that affect their lives. Yet, in the multi-agent systems we build today, this truth is too often forgotten \cite{joseph2016fairness, hossain2021fair, jones2023efficient}. Fairness is almost always reduced to optimizing for a specific outcome: the sum of utilities, the balancing of welfare, or the smoothing of inequality, echoing the consequentialist tradition of judging actions by their aggregate results \cite{sep-consequentialism, smart1973utilitarianism}. While these notions of fairness may provide elegance and tractability, they miss the very essence of what we consider to be fair. They consider what is gained by a set of decisions, not how they were decided. This is an imposition of values from outside of the system, rather than respecting the very agency from within.

This paper begins from the conviction that fairness in multi-agent systems must be grounded not in optimal outcomes, but in the principle of equal voice. To guarantee this is to honour the dignity of participation in multi-agent decision-making; to ignore it is to risk building systems that sacrifice legitimacy, that is, whether the decisions themselves are perceived as rightful and acceptable, for the sake of efficiency. This principle reflects a contractualist view of fairness, which holds that a decision is only legitimate if it cannot be reasonably rejected by those subject to it \cite{scanlon1998we}. We call this principle procedural fairness.

This moral insight is not merely philosophical. Extensive evidence in psychology and economics shows that people consistently value fair process, even if that means outcomes are less than ideal \cite{ANAND2001247, lind2013social, tyler1990people}. For example, \citet{lind2013social} recount an example of a woman whose traffic ticket case was dismissed; however, she still left the courtroom angry because she felt she had compelling evidence and the judge never heard her argument. In fact, many people reported feeling the same way, despite being handed the best possible outcome from the court. This same dynamic appears in collective allocation systems like participatory budgeting, where communities not only want good projects but a voice in which projects are chosen, or may only want to contribute to projects of which they approve, a concept referred to as decomposability in the literature.

Yet existing approaches in multi-agent learning overwhelmingly reduce fairness to outcomes such as utilitarianism, Nash welfare, or inequality. While they capture important values, they impose fairness as an external criterion rather than letting it arise from the agents themselves. What is missing in the literature is a framework that gives agents themselves an equal share of decision-making power. Inspired by Rawls' notion of \emph{pure procedural justice} \cite{41832fbd-4a82-3c18-ae99-e02c0f77c64f}, we formalize procedural fairness in MA-MABs, a framework where each action (pulling an arm) produces potentially different rewards for each agent, sampled from potentially different distributions. This framework naturally captures both the allocation of benefits and the distribution of decision-making power in a simple and easy-to-understand way.

To situate procedural fairness, we compare it against two outcome-based objectives: \emph{inequality minimization}, where outcomes are distributed so that agents receive outcomes that are as equal as possible, and \emph{utilitarianism}, which maximizes aggregate welfare. We also benchmark against other fairness objectives, such as egalitarian welfare, Nash welfare, and generalized Gini welfare, which balance these concerns in different ways.


Our central claim is that procedural fairness deserves recognition alongside traditional notions of fairness like Nash welfare, inequality, and utilitarianism, not as an alternative, but as a principle of legitimacy. Our technical contributions include:

\begin{itemize}
\item A formal definition of procedural fairness for MA-MABs.
    \item Impossibility results showing procedural fairness and outcome-based fairness are fundamentally incompatible, implying that fairness requires normative choices.
    \item An algorithm for learning a procedurally fair policy with sublinear regret guarantees.

    \item Empirical evaluation across multiple settings illustrating that procedural fairness balances efficiency and equality while preserving legitimacy.
\end{itemize}

\noindent{\bf Related Work}

Our work draws on established insights from psychology, economics, and political philosophy holding that process fairness is distinct from outcome fairness. In psychology, Tyler and Lind \citep{lind2013social, tyler1990people} demonstrate that individuals inherently value having a voice, with empirical evidence confirming a preference for fair processes over purely optimal outcomes \citep{ANAND2001247}. In political philosophy, Rawlsian pure procedural justice \citep{41832fbd-4a82-3c18-ae99-e02c0f77c64f} and Scanlon's contractualism \citep{scanlon1998we} locate legitimacy in the justifiability of decision-making rules themselves.

Within the bandit literature, existing fairness formulations broadly focus on arm-centric constraints or private item allocation rather than procedural voice in public goods. Arm-centric approaches guarantee minimum pull rates \citep{patil2021achieving} or enforce merit-based selection \citep{joseph2016fairness}. Online fair division sequentially allocates private items under envy-freeness or proportionality constraints \citep{procaccia2024honor, schiffer2025improved}, contrasting with our setting,  which is more analogous to a public-good setting, where an arm choice yields rewards for all agents. The most similar setting is MA-MABs with public rewards, where arm pulls yield $N$-dimensional reward vectors. However, here too the literature focuses primarily on utilitarian coordination \citep{chakraborty2017coordinated} or outcome-based fairness metrics---such as Nash Social Welfare via product ($NSW_{prod}$) \citep{hossain2021fair, jones2023efficient} or geometric mean \citep{zhang2024no}, the Generalized Gini Index \citep{busa2017multi}, and Pareto optimality \citep{xu2023pareto}. In contrast, our framework incorporates procedural representation directly.

Our formulation is related to participatory budgeting, particularly the notion of decomposability, as well as donor coordination, and random dictatorship: with known means, it resembles budget aggregation \citep{BRANDL2022102585,aziz2024fairlotteriesparticipatorybudgeting, BOGOMOLNAIA2005165}. Our contribution lies in proposing procedural fairness for the bandit setting, as well as the broader philosophical point that representation is a valuable, and often overlooked, objective in its own right. 

%% file: sections/multi-agent-bandits.tex
\section{Multi-Agent Multi-Armed Bandits}

Let $N \geq 2$ be the number of agents and $K \geq 2$ the number of arms, and let $[N] := \{1, \dots, N\}$ and $[K] := \{1, \dots, K\}$. A \textit{policy} is a distribution $P = (p_1, \dots, p_K) \in \Delta^K$, where $p_k$ is the probability that arm $k$ is pulled and $\Delta^K := \{p \in \mathbb{R}_{\ge 0}^K : \sum_{k=1}^K p_k = 1\}$ is the probability simplex over the arms.

When an arm is pulled, all agents receive some reward drawn from a distribution that stays fixed across time steps. Agents will not necessarily receive the same reward, and distributions may vary from agent to agent and from arm to arm. We let $\mu^* \in \mathbb{R}^{N \times K}$ represent the agents' true reward means, where $\mu^*_{i,k}$ represents the mean reward agent $i$ receives when arm $k$ is pulled. Additionally, let $\hat{\mu}^t$ denote the agents' reward estimates at time $t$, where $\hat{\mu}^t_{i,k}$ is the estimate at time $t$ of the reward agent $i$ receives when arm $k$ is pulled. 
Finally, for a mean-reward matrix $\mu$, define agent $i$'s set of favourite arms as $F_i(\mu) := \{ j \in [K] \mid \mu_{i,j} = \max_{k \in [K]} \mu_{i,k} \}$, and write $F_i := F_i(\mu^*)$ for the true favourite-arm set.

In all instances, we assume that each true reward mean is bounded between 0 and 1, i.e., $\mu^*_{i,k} \in [0,1]$ for all $i \in [N]$ and $k \in [K]$, and that drawn rewards are in $[0, 1]$ from distributions fixed across time steps, with rewards i.i.d.\ across pulls for each agent-arm pair.

We now define the following concepts:

\begin{definition}[Utility]
    The utility of an agent $i$ under a policy $P = (p_1, \dots, p_K)$ is defined as $U_i(P) = \sum_{k=1}^{K} p_k \mu^*_{i,k}$.
\end{definition}

\begin{definition}[Decision Share] The decision share of agent \( i \) under a policy \( P = (p_1, \dots, p_K) \) is the total probability assigned to the agent's favourite arm(s): $\beta_i(P) = \sum_{k \in F_i} p_k$.
\end{definition}

\begin{definition}[Utility-Based Nash Welfare]
    The utility-based Nash welfare under a policy $P$ is defined as $\NWU(P) := \prod_{i=1}^N U_i(P) = \prod_{i=1}^N \left( \sum_{k=1}^K p_k \mu^*_{i,k} \right)$.
\end{definition}

\begin{definition}[Decision-Share-Based Nash Welfare]
    The decision-share-based Nash welfare under a policy $P$ is $\NWbeta(P) := \prod_{i=1}^N \beta_i(P) = \prod_{i=1}^N \left( \sum_{k \in F_i} p_k \right)$. In other words, it is the Nash welfare of the agents' decision shares.
\end{definition}

%% file: sections/procedural-fairness.tex
\section{Procedural Fairness in Multi-Agent Multi-Armed Bandits}



Procedural fairness distinguishes representation from outcome quality. Outcome-based objectives (like Nash welfare or utilitarianism) can yield high expected utility for an agent without ever selecting their top choice. Consider two children ranking car seats: front, back-right, then back-left. If one child always gets back-right while the other alternates between front and back-left, an outcome-based metric may be satisfied, yet only one child ever receives the front seat. The first child is compensated with a good second choice, but lacks equal representation. Favourable lower-ranked outcomes may compensate for utility, but cannot substitute for equal influence over decisions.

We therefore model influence as equal and non-transferable. For $N$ agents, equal influence requires partitioning the policy mass into $N$ equal $1/N$ shares, with each agent's share assigned exclusively to their favourite arms. This design choice mirrors the single non-transferable vote, where each voter's support is non-transferable and counts only toward their top choice. While traditional models might suggest agents redirect their shares toward a mutually acceptable non-favourite arm to boost expected rewards, procedural fairness denies that such compromises are equivalent. Utility, $U_i(P)$, measures outcome value, not the intrinsic value of having preferences represented. Consequently, any inefficiency arising from preserving representation is not a defect, but a direct result of treating procedural legitimacy as an objective in its own right.

While procedural fairness in MA-MABs can be implemented via mechanisms like random dictatorship \citep{gibbard77}, we focus on learning a specific probability distribution over arms. Here, equal voice is instantiated at the policy level: rather than selecting agents directly, a policy is procedurally fair if its arm-selection probabilities can be decomposed into equal, preference-consistent agent shares. Specifically:

\begin{definition}[Procedural Fairness in MA-MABs]
    \label{df:proceduralfairness}
    Let \(P\) be a policy, and \(y_{i,k}\geq 0\) denote the probability mass that agent \(i\) contributes to arm \(k\). A policy $P$ is procedurally fair if the probability distribution can be attributed to agents according to the following conditions:

    \par\hangindent=1em\textbf{1. Equal decision-making influence.}
    Each agent \( i \in [N] \) is allocated an equal share of the total probability mass, $\sum_{k=1}^{K} y_{i,k} = \frac{1}{N}$ for every $i \in [N]$, where \( y_{i,k} \) represents the probability mass that agent \( i \) contributes to selecting arm \( k \).

    \par\hangindent=1em\textbf{2. Preference-based allocation.}
    Each agent's attributed probability mass is assigned to their most-preferred arm(s), defined as the set of arms with the highest mean reward for that agent: $F_i = \{ j \in [K] \mid \mu^*_{i,j} = \max_{k \in [K]} \mu^*_{i,k} \}$. If multiple arms achieve the same maximum expected reward, an agent may distribute their probability mass arbitrarily among them. Specifically, $y_{i,k} = 0$ for every $i \in [N]$ and every $k \notin F_i$.

    \par\hangindent=1em\textbf{3. Policy consistency.}
    The probability assigned to each arm equals the total probability mass contributed to that arm by all agents: $p_k = \sum_{i=1}^{N} y_{i,k}$ for every $k \in [K]$, where \(y_{i,k} \geq 0\) for every agent \(i\) and arm \(k\).
\end{definition}

To measure how procedurally fair a given policy is (which we call a \emph{score}), we formulate an optimization problem. The intuition is: given some probability distribution, can we allocate $\frac{1}{N}$ of probability on behalf of each agent on their favourite arms, subject to the given policy? The extent to which we can allocate the attributed probability mass is the procedural fairness score. We define $y_{i,j}$ as an allocation variable, representing how much of agent $i$'s decision share is allocated to arm $j$, provided that arm $j \in F_i(\mu)$.  The optimization, which we call $\PF(\mu, P)$, is as follows:
\begin{align*}
\max \quad &\sum_{i=1}^{N} \sum_{j \in F_i(\mu)} y_{i,j} \\
\text{subject to} \quad &\sum_{j \in F_i(\mu)} y_{i,j} \leq \frac{1}{N} \quad \forall i \in [N] \\
&\sum_{i: j \in F_i(\mu)} y_{i,j} \leq p_j \quad \forall j \in [K] \\
&y_{i,j} = 0 \quad \forall i \in [N],\ \forall j \notin F_i(\mu), \; y_{i,j} \geq 0 \; \forall i, j
\end{align*}

The score lies in $[0,1]$ and $\PF(\mu^*, P) = 1$ if and only if $P$ is procedurally fair. The following example illustrates the score.

\input{proofs/examples-procedural-score}

%% file: proofs/examples-procedural-score.tex
\subsection{Illustrative Procedural-Fairness Score}
\label{sec:pf-score-example}
To illustrate, consider the utility matrix
\[
\mu =
\begin{bmatrix}
1 & 0 \\
1 & 0 \\
0 & 1 \\
\end{bmatrix}
\]
and policy $P=(\frac{1}{2}, \frac{1}{2})$.
We aim to find $y_{i,j}$ for each agent $i$ and arm $j$, distributing each agent's $\frac{1}{3}$ share of the total probability mass to their favourite arm without exceeding any arm's total availability in the policy $P = (\frac{1}{2}, \frac{1}{2})$.

Agent 1 prefers arm 1. Allocate their full share: $y_{11} = \frac{1}{3}$, $y_{12} = 0$. Then, agent 2 also prefers arm 1. Since $\frac{1}{3}$ is already used by agent 1 with arm 1, only $\frac{1}{6}$ remains (recall that $y_{11} + y_{21} \leq \frac12$ must hold, and we already have $\frac13 + y_{21} \leq \frac12$ from agent 1 allocating their share of the probability mass to arm 1). So: $y_{21} = \frac{1}{6}$, $y_{22} = 0$. Finally, agent 3 prefers arm 2. We can fully allocate their share of the probability mass to arm 2, so: $y_{32} = \frac{1}{3}$, $y_{31} = 0$. Summing all $y_{i,j}$ gives $\frac{5}{6}$, the procedural fairness score. This reflects agent-level allocation ratios of 1, 0.5, and 1, averaging to $\frac{5}{6}$.

\subsection{Proportionality Guarantee}

Procedurally fair policies also have additional guarantees with respect to the total amount of decision share that each agent will receive:

\begin{proposition}
\label{thm:RRimpliesProp}
A procedurally fair policy gives agents at least $1/N$ of their maximum decision share \textit{and} at least $1/N$ of their maximum achievable utility, in expectation.
\end{proposition}

\input{proofs/pf-guarantees}

%% file: proofs/pf-guarantees.tex
\begin{proof}
    The maximum decision share an agent can obtain is 1, attained by placing all probability on their favourite arm(s). Since a procedurally fair policy guarantees each agent at least $1/N$ probability on their favourite arm(s), each agent receives at least $1/N$ of their maximum decision share. Similarly, agent $i$'s maximum achievable expected utility is $\max_{k} \mu^*_{i,k}$, attained by placing all probability on a favourite arm. Because a procedurally fair policy assigns at least $1/N$ probability to favourite arms, it gives agent $i$ at least $1/N$ of this maximum utility in expectation, thus giving them a proportional share. This is similar to the concept of \emph{proportionality} \citep{conitzer2017fair}.
\end{proof}

%% file: sections/learning-procedural-fairness.tex
\section{Learning Procedural Fairness}

To learn a procedurally fair policy, we formulate an objective that ensures each agent allocates an equal amount of the probability mass to their most preferred arms. This formulation is then used in our algorithm for learning a procedurally fair policy (which we describe later in the paper). Additionally, recall that $\hat{\mu}^t_{i,k}$ is the estimated mean of arm $k$ for agent $i$ at time step $t$. The procedurally fair policy is selected by maximizing the decision-share-based Nash welfare:
\[
\max_{p \in \Delta^K} \prod_{i=1}^N \sum_{j \in F_i} p_j.
\]

Although this objective does not explicitly impose procedural fairness constraints, we show that maximizing this objective results in procedurally fair policies. Throughout this work, we adopt the Nash formulation because it selects policies that satisfy additional theoretical properties described later. In particular, we first show that this objective is procedurally fair:

\begin{theorem}
\label{thm:dsbnw-is-pf}
    A policy selected by maximizing the decision-share-based Nash welfare objective is a procedurally fair policy.
\end{theorem}

\input{proofs/dbnw-is-procedurally-fair}

\subsection{Learning Algorithm}

\begin{algorithm}[t]
\caption{Learning a Procedurally Fair Policy}
\label{alg:learn_pf}
\begin{algorithmic}[1]

\REQUIRE Rounds $T$, number of agents $N$, number of arms $K$,
exploration decay $\gamma$, confidence multiplier $\alpha \geq 1$

\STATE Initialize
$P^0_{\PF} \gets \frac{1}{K}\mathbf{1}_K$

\FOR{$k = 1$ to $K$}
    \STATE Pull arm $k$ and observe $(r_i^k)_{i=1}^N$
\ENDFOR
\STATE Set $n_k^{K+1} \gets 1$, $\hat{\mu}_{i,k}^{K+1} \gets r_i^k$ for all $k \in [K]$ and $i \in [N]$

\FOR{$t = K+1$ to $T$}
    \STATE Draw $E_t \sim \mathrm{Bernoulli}\!\left(t^{-(1-\gamma)}\right)$
    \STATE $z_k^t \gets \sqrt{\frac{2\ln(NKt)}{n_k^{t}}}$ for all $k \in [K]$

    \FOR{each agent $i \in [N]$}
        \STATE $j_i^t \gets \arg\max_{j\in[K]} \hat{\mu}_{i,j}^{t}$
        \STATE $\hat{F}_i^t \gets \left\{k \in [K]:
        \hat{\mu}_{i,k}^{t}+\alpha z_k^t
        \geq
        \hat{\mu}_{i,j_i^t}^{t}-\alpha z_{j_i^t}^t\right\}$
    \ENDFOR

    \STATE $P^t_\PF \gets \arg\max_{p\in\Delta^K} \prod_{i=1}^{N}\sum_{k\in\hat{F}_i^t}p_k$

    \IF{$E_t = 1$}
        \STATE Sample $\tilde{k}_t$ uniformly from $[K]$ and set $P_t \gets e_{\tilde{k}_t}$
    \ELSE
        \STATE $P_t \gets P^t_\PF$
    \ENDIF

    \STATE Sample $k_t \sim P_t$, pull arm $k_t$, and observe $(r_i^t)_{i=1}^N$
    \STATE $n_{k_t}^{t+1} \gets n_{k_t}^{t}+1$
    \STATE $\hat{\mu}_{i,k_t}^{t+1} \gets \hat{\mu}_{i,k_t}^{t}
    +\frac{r_i^t-\hat{\mu}_{i,k_t}^{t}}{n_{k_t}^{t+1}}$
    for all $i \in [N]$
    \STATE $n_{k}^{t+1} \gets n_{k}^{t}$, $\hat{\mu}_{i,k}^{t+1} \gets \hat{\mu}_{i,k}^{t} \; \forall k \neq k_t, \; \forall i \in [N]$
\ENDFOR

\RETURN $P_{\PF}$

\end{algorithmic}
\end{algorithm}

We now introduce our algorithm for learning a procedurally fair policy, Algorithm~\ref{alg:learn_pf}. When learning the optimal policy, the favourite set is derived using UCB-style concentration bounds. For each agent, we compare every arm's upper confidence bound with the lower confidence bound of the empirically best arm. Any arm whose UCB overlaps this lower bound remains in the favourite set. To guarantee convergence, we must ensure that these intervals shrink over time, as we need the intervals to converge to 0 to recover the true favourite set. To solve this problem, we borrow a common mechanism from the literature \citep{auernonstochastic} and select an arm at random with probability $t^{-(1-\gamma)}$, where $\gamma \in (0, 1)$ is a decay parameter. This guarantees that every arm is pulled sufficiently often so that the confidence radius vanishes as $t \rightarrow \infty$.

\subsection{Regret Analysis}

Learning a procedurally fair policy requires identifying each agent's favourite arms to enforce equal influence. The correctness of procedural fairness depends on identifying each agent's favourite set and enforcing equal influence in the resulting policy. For this reason, we define PF regret in terms of the number of mismatch rounds in favourite-set recovery, which is analogous to a mistake-bound criterion.

\noindent
\textbf{Procedural Fairness Regret.}
$R^{\PF}(T) := K + \sum_{t=K+1}^T \mathbf{1}\{\exists i \in [N] : \hat{F}_i^t \neq F_i\}$, counting mismatch rounds between the estimated candidate sets $\hat{F}_i^t$ and the true favourite-arm sets $F_i$. For the first $K$ rounds we do forced exploration, so we define regret for these rounds to be $K$.

The analysis has three steps. We first establish a uniform concentration event, then lower-bound the number of samples obtained for every arm through randomized exploration, and finally bound the time required to recover every agent’s true favourite set.

\input{proofs/hoeffding-concentration}

\input{proofs/uniform-hoeffding-concentration}

\input{proofs/randomized-exploration-bound}

\input{proofs/arm-counts-adjustment-bounds}
\input{proofs/exploitation-phase-bound}

Tying together our previous results, we can now present our final bound.

\begin{theorem}
\label{thm:pf_regret}
Let $\alpha \geq 1$ and $\Delta_{\min} :=\min_{i:\, |F_i| < K}\ \min_{j \in F_i}\ \min_{k \notin F_i}\ \bigl(\mu^*_{i,j} - \mu^*_{i,k}\bigr)$, where $\Delta_{\min} = \infty$ when $|F_i| = K$ for every agent $i$. Then, with high probability, the regret bound for the Procedural Fairness algorithm is $O\left(K+\frac{T^\gamma}{\gamma} +\left(\frac{(1+\alpha)^2\gamma K \ln{(NKT)}}{\Delta_{\min}^2}\right)^{\frac{1}{\gamma}}\right)$.
\end{theorem}
\input{proofs/full-pf-regret-bound}

%% file: proofs/dbnw-is-procedurally-fair.tex
\begin{proof}
    Let $P^*=(p^*_1, \dots, p^*_K)$ be a policy that maximizes decision-share-based Nash welfare. Specifically, it maximizes the following objective: $\prod^N_{i=1} \sum_{k\in F_i} p_k$. Define $d_i = \beta_i(P) = \sum_{k \in F_i} p_k$. Therefore, the objective can be re-written as $\prod^N_{i=1} d_i$, which is equivalent to maximizing $L(P) =  \sum_{i=1}^N \log d_i$.

    We know a procedurally fair policy always exists: give each agent $1/N$ probability on any of their favourite arms. Therefore, the optimal Nash product is strictly positive, and $d_i>0$ for every agent $i$.

    For each arm $k$, define its marginal Nash score as $s_k = \sum_{i:k\in F_i}\frac{1}{d_i}$. Consider two arms $k$ and $\ell$, where $p^*_\ell > 0$ and $p^*_k > 0$. If $s_k > s_\ell$, then moving a sufficiently small amount of probability from arm $\ell$ to arm $k$ would increase the log Nash welfare. Specifically, choose an $\epsilon$ such that $p^*_\ell > \epsilon > 0$. The new policy, $P(\epsilon)$, can be defined as $p_k(\epsilon) = p^*_k + \epsilon$ and $p_\ell(\epsilon) = p^*_\ell - \epsilon$, and every other probability is unchanged. This remains a valid policy since $p^*_\ell - \epsilon \geq 0$ and $(p^*_\ell - \epsilon) + (p^*_k + \epsilon) = p^*_\ell + p^*_k$, meaning the total probability remains unchanged. Under this new policy, agent $i$'s decision share becomes $d_i(\epsilon) = d_i + \epsilon \mathbf{1}\{k \in F_i\} - \epsilon \mathbf{1} \{\ell \in F_i\}$. The derivative of $\log d_i(\epsilon)$ w.r.t $\epsilon$ at $\epsilon = 0$ is equal to $\frac{1}{d_i} (\mathbf{1} \{k \in F_i\} - \mathbf{1} \{\ell \in F_i\})$. When we sum over all agents, we get the derivative of $L(P(\epsilon))$ w.r.t $\epsilon$ at $\epsilon = 0$ to be $\sum_{i=1}^N \frac{\mathbf{1} \{k \in F_i\}}{d_i} - \sum_{i=1}^N \frac{\mathbf{1} \{\ell \in F_i\}}{d_i}$. Since $s_k = \sum_{i:k \in F_i} \frac{1}{d_i}$, the derivative is exactly $s_k - s_\ell$. So if $s_k > s_\ell$, that means that the derivative of $L(P(\epsilon))$ w.r.t $\epsilon$ at $\epsilon = 0$ is greater than zero, which means that $L(P(\epsilon)) > L(P^*)$, which is a contradiction. Therefore, whenever $p^*_\ell > 0$, $s_k \leq s_\ell$. Moreover, since $k$ and $\ell$ both receive positive probability, applying that same result with the roles flipped gives $s_\ell \leq s_k$. Therefore, $s_k = s_\ell$, meaning every arm that has positive probability must have the same value.

    Let that value be denoted as $\lambda$. Then:
    {\allowdisplaybreaks
    \begin{align*}
    \lambda
    &= s_k \\
    &= s_k \sum_{a=1}^K p_a^*
        \qquad \text{(recall $\sum_{a=1}^K p_a^* = 1$)}\\
    &= \sum_{a=1}^K p_a^* s_a\\
    &= \sum_{a=1}^K
        p_a^*
        \sum_{i:a\in F_i}\frac{1}{d_i}\\
    &= \sum_{\substack{i,a\\ a \in F_i}}
        \frac{p_a^*}{d_i}\\
    &= \sum_{i=1}^N
        \sum_{a \in F_i}\frac{p_a^*}{d_i}\\
    &= \sum_{i=1}^N
        \frac{1}{d_i}
        \sum_{a \in F_i}p_a^*\\
    &= \sum_{i=1}^N\frac{d_i}{d_i}\\
    &= N.
    \end{align*}
    }

    Therefore, $s_k = N$ whenever $p^*_k > 0$.

    Now we define the individual allocation:
    \[
    y_{i,k}
    =
    \begin{cases}
    \dfrac{p_k^*}{N d_i}, & k\in F_i,\\[6pt]
    0, & k\notin F_i.
    \end{cases}
    \]

    For each agent $i$, we know that
    \[
    \begin{aligned}
    \sum_{k=1}^K y_{i,k}
    &=
    \frac{1}{N d_i}
    \sum_{k\in F_i}p_k^*\\
    &=
    \frac{d_i}{N d_i}\\
    &=
    \frac1N.
    \end{aligned}
    \]
    Therefore, every agent controls exactly $1/N$ of the policy mass, and by definition of $y_{i,k}$, only assigns it to their favourite arms.

    For each arm $k$ with $p^*_k > 0$,

    \[
    \begin{aligned}
    \sum_{i=1}^N y_{i,k}
    &=
    \frac{p_k^*}{N}
    \sum_{i:k\in F_i}\frac{1}{d_i}\\
    &=
    \frac{p_k^*}{N}s_k\\
    &=
    \frac{p_k^*}{N}N\\
    &=
    p_k^*.
    \end{aligned}
    \]

    If $p^*_k = 0$, then $y_{i,k} = 0$ for every agent, so $\sum_i y_{i,k} = p^*_k = 0$. Therefore, the individual allocations aggregate exactly to $P^*$. Hence $P^*$ is a procedurally fair policy.
\end{proof}

%% file: proofs/hoeffding-concentration.tex
We first establish that, at any fixed horizon, every empirical agent-arm mean is simultaneously close to its true value, providing the basic concentration guarantee needed to reason about favourite-set recovery.

\begin{lemma}[Hoeffding Mean Concentration]
\label{lm:meanconc}
Let $z^\ell_k = \sqrt{\frac{2\ln{(NKt)}}{n^\ell_k}}$, where $n^\ell_k$ represents the number of times arm $k$ has been pulled by time step $\ell$. Then, with probability at least $1 - \frac{2}{(NKt)^3}$, we have that $\forall i \in [N], k \in [K], \ell \in [K+1,...,t]: |\hat{\mu}^\ell_{i,k} - \mu^*_{i,k}| \leq z^\ell_k$
\end{lemma}
\begin{proof}
    This proof is nearly identical to an existing proof \cite{hossain2021fair}. Consider a fixed time step $t$. Let $X_{i, k}^m$ denote the rewards received by agent $i$ from the first $m \in [t]$ pulls of arm $k$. Further, let $z^m_k = \sqrt{\frac{2\ln{(NKt)}}{m}}$. Let $\bar{\mu}^m_{i,k} = 1/m \sum_{x \in X^m_{i,k}} x$ be the average reward received by agent $i$ from the first $m$ pulls of arm $k$. Recall that each reward $x$ is in $[0, 1]$. Then, by Hoeffding's inequality followed by a union bound, we have:

    \[
    \begin{aligned}
        \forall i \in [N],\, k \in [K],\, m \in [t]: P\!\left(\Big|\sum_{x \in X^m_{i,k}} x - m \mu^*_{i,k}\Big| > m z^m_k\right)
        &= P\left(\big| m \bar{\mu}^m_{i,k} - m \mu^*_{i,k}\big| > m z^m_k\right) \\[4pt]
        &= P\left( \big| \bar{\mu}^m_{i,k} - \mu^*_{i,k} \big| > z^m_k \right) \\[4pt]
        &\leq 2\exp\left( -\frac{2(mz^m_k)^2}{m} \right) \\[4pt]
        &= 2\exp\left(-2m(z^m_k)^2\right) \\[4pt]
        &= 2\exp\left(-2m\left(\sqrt{\frac{2\ln(NKt)}{m}}\right)^2\right) \\[4pt]
        &= 2\exp\left(-4\ln(NKt)\right) \\[4pt]
        &= \frac{2}{(NKt)^4}
    \end{aligned}
    \]

    By the union bound, we get
    \[
    \begin{aligned}
    \Pr\!\Big[\exists\, i\in[N],\,k\in[K],\,m\in[t]:\
    \bigl|\bar{\mu}^m_{i,k}-\mu^*_{i,k}\bigr| > z^m_k\Big]
    &\le \sum_{i\in[N]}\sum_{k\in[K]}\sum_{m\in[t]} \frac{2}{(NKt)^4} \\[4pt]
    &= \frac{2\,N K t}{(NKt)^4} \\[4pt]
    &= \frac{2}{(NKt)^3}
    \end{aligned}
    \]

    Taking the complement, we get
    \[
    \Pr\!\Big[\forall\, i\in[N],\,k\in[K],\,m\in[K+1,...,t]:\
    \bigl|\bar{\mu}^m_{i,k}-\mu^*_{i,k}\bigr| \le z^m_k\Big]
    \;\ge\; 1 - \frac{2}{(NKt)^3}
    \]

    Since for $\ell \in [K+1,...,t]$, we know that $1 \leq n^\ell_k \leq \ell - 1 \leq t$ and the above holds for every possible $m \in [1,...,t]$, it must also hold for the particular pull count $n^\ell_k = m$. Therefore, our above inequality holds for $z^\ell_k = \sqrt{\frac{2\ln(NKt)}{n^\ell_k}}$.

\end{proof}

%% file: proofs/uniform-hoeffding-concentration.tex
We then extend this fixed-horizon guarantee to a single high-probability event that holds uniformly over all rounds, allowing the subsequent analysis to treat the confidence intervals as valid throughout the entire run.

\begin{corollary}[Uniform-in-time Hoeffding Mean Concentration]
\label{cor:meanconc-uniform}
Fix a horizon $T$ and let $z^\ell_k(t) = \sqrt{\frac{2\ln(NKt)}{n^\ell_k}}$.
Let $\mathcal{E}_T$ denote the event that
$|\hat{\mu}^\ell_{i,k}-\mu^*_{i,k}|\le z^\ell_k(t)$ for every
$t\in\{K+1,\dots,T\}$, $i\in[N]$, $k\in[K]$, and $\ell\le t$. Then
\[
\Pr(\mathcal{E}_T) \ge 1-\frac{3}{(NK)^3}.
\]
\begin{proof}
    By Lemma \ref{lm:meanconc} and union-bounding over $t=K+1,\dots,T$,
    \[
    \begin{aligned}
    \Pr(\mathcal{E}_T)
    &\ge 1-\frac{2}{(NK)^3}\sum_{t=K+1}^{T}\frac{1}{t^3}\\
    &\ge 1-\frac{2}{(NK)^3}\sum_{t=1}^{T}\frac{1}{t^3}.
    \end{aligned}
    \]
    Applying the integral bound for decreasing functions,
    \[
    \sum_{t=1}^{T}\frac{1}{t^3} \;\le\; 1+\int_{1}^{T} x^{-3}\,dx
    = \tfrac{3}{2}-\tfrac{1}{2T^{2}},
    \]
    \[
    \begin{aligned}
    \Pr(\mathcal{E}_T)
    &\ge 1-\frac{2}{(NK)^3}
    \Bigl(\tfrac{3}{2}-\tfrac{1}{2T^{2}}\Bigr)\\
    &\ge 1-\frac{3}{(NK)^3}.
    \end{aligned}
    \]
\end{proof}
\end{corollary}

%% file: proofs/randomized-exploration-bound.tex
Next, we show that randomized exploration occurs often enough to support learning but only sublinearly many times, so its direct contribution to regret remains controlled.

\begin{lemma}
\label{lm:randbound}
    Let $E_t$ be independent Bernoulli random variables with $\Pr(E_t = 1) = q_t = t^{-(1-\gamma)}$, and let $S_T = \sum_{t=K+1}^{T}E_t$. Then, with high probability, $S_T = O\left(\frac{T^\gamma}{\gamma}\right)$, and since each exploration round can only incur at most one favourite-set mismatch, the number of mismatches occurring on exploration rounds is also $O(\frac1\gamma T^\gamma)$.
\end{lemma}

\begin{proof}
    By definition, we have that $S_T = \sum_{t=K+1}^{T}E_t$. Because $E_t$ is Bernoulli with mean $q_t$, we know that $\mathbb{E}[E_t] = q_t$. By linearity of expectation, we have:

    \[
    \mathbb E[S_T] = \mathbb E\left[\sum_{t=K+1}^{T}E_t\right] =
    \sum_{t=K+1}^{T}\mathbb E[E_t] = \sum_{t=K+1}^{T}q_t
    \]

    Since $q_t = t^{-(1-\gamma)}$, we have that $\mathbb E[S_T] = \sum_{t=K+1}^T t^{-(1-\gamma)}$. We can find the bound by taking the integral:

    \[
    \int_{K+1}^{T+1} x^{-(1-\gamma)}\,dx
    \leq
    \sum_{t=K+1}^T t^{-(1-\gamma)}
    \leq
    \int_K^T x^{-(1-\gamma)}\,dx.
    \]

    This gives:

    \[
    \frac{(T+1)^\gamma - (K+1)^\gamma}{\gamma}
    \leq
    \sum_{t=K+1}^T t^{\gamma-1}
    \leq
    \frac{T^\gamma - K^\gamma}{\gamma}
    \]

    This gives the expected number of exploration rounds as $\mathbb E[S_T] = O(\frac1\gamma T^\gamma)$. We now want to bound $S_T$ itself.

    We know from a Chernoff bound that $\Pr(S_T \geq 2 \mathbb E[S_T]) \leq e^{- \mathbb E[S_T]/3}$. Therefore, with probability at least $1-e^{-\mathbb E[S_T]/3}$, we have that $S_T \leq 2 \mathbb E[S_T]$. Since we know that $\mathbb E[S_T] = O(\frac1\gamma T^\gamma)$, we know that $S_T = O(\frac1\gamma T^\gamma)$.

\end{proof}

%% file: proofs/arm-counts-adjustment-bounds.tex
We then prove that randomized exploration gives every arm sufficiently many samples, which forces all confidence radii to shrink at a quantifiable rate over time.

\begin{lemma}
\label{lm:ucbdecay}
Fix a horizon $T$ and a deterministic $K+1 \leq \tau \leq T$. At each round $s \in \{K+1,...,T\}$, the procedural fairness algorithm explores with probability $p_{\mathrm{rand}}^s=s^{-(1-\gamma)}$, with $0 < \gamma < 1$, and selects an arm uniformly at random. Define $m_\tau := 1+\frac{1}{K}\sum_{s=K+1}^{\tau-1}p_{\mathrm{rand}}^s = 1+\frac{1}{K}\sum_{s=K+1}^{\tau-1}s^{-(1-\gamma)}.$ Then:

\begin{enumerate}
    \item For every arm $k$, $\mathbb{E}[n^\tau_k] \geq m_\tau$.
    \item With probability at least $1 - K\exp\left(-\frac{m_\tau}{8}\right)$, $n^t_k \geq m_\tau/2$ for every arm $k$ and every $t \in \{\tau,...,T\}$.
    \item On the same event, for every arm $k$ and every $t \in \{\tau,...,T\}$ simultaneously:
    \[
        z^t_k = \sqrt{\frac{2\ln(NKt)}{n_k^t}}
        \le \sqrt{\frac{4\ln(NKT)}{m_\tau}}
        \le \sqrt{\frac{4\gamma K\ln(NKT)} {\tau^\gamma-1}}.
    \]
\end{enumerate}
\end{lemma}
\begin{proof}

Let $Q_{s,k} = \mathbf{1}\{\text{round $s$ is exploratory and selects arm $k$}\}$, which are independent Bernoulli random variables for each arm $k$. Let $Y^\tau_k = 1+\sum_{s=K+1}^{\tau-1} Q_{s,k}$, where the leading $1$ represents the forced initialization pull of arm $k$. We also know that $\Pr(Q_{s,k} = 1) = p^s_{\mathrm{rand}}/K$. Therefore,

\[
\mathbb E[Y_k^\tau]
= 1+\sum_{s=K+1}^{\tau-1}\mathbb E[Q_{s,k}]
= 1+\frac{1}{K}\sum_{s=K+1}^{\tau-1}p_{\mathrm{rand}}^s
= m_\tau.
\]

Because $n^\tau_k$ counts all pulls of arm $k$ available by timestep $\tau$, including the forced initialization pull and exploitation pulls, $n^\tau_k \geq Y^\tau_k$. Consequently, $\mathbb{E}[n^\tau_k] \geq \mathbb E[Y^\tau_k] = m_\tau$, proving the first claim.

Since the leading $1$ in $Y^\tau_k$ can be viewed as a Bernoulli random variable with success probability $1$, a Chernoff bound gives:

\[
\Pr \left(Y^\tau_k < \frac12 \mathbb{E}[Y^\tau_k]\right) = \Pr \left(Y^\tau_k < \frac{m_\tau}{2}\right) \leq \exp\left(-\frac{m_\tau}{8}\right)
\]

Since $n_k^\tau\ge Y_k^\tau$,
\[
\Pr\left(
n_k^\tau<\frac{m_\tau}{2}
\right)
\le
\exp\left(-\frac{m_\tau}{8}\right).
\]

Taking a union bound over the $K$ arms gives
\[
\Pr\left(
\forall k\in[K]:
n_k^\tau\ge\frac{m_\tau}{2}
\right)
\ge
1-K\exp\left(-\frac{m_\tau}{8}\right).
\]

On this event, arm counts cannot decrease, so for every $t \in \{\tau,...,T\}$ we know that $n^t_k \geq n^\tau_k \geq m_\tau / 2$, proving the second claim.

Furthermore, for every $\tau \leq t \leq T$, we know that $\ln(NKt) \leq \ln(NKT)$. Combining with the preceding count gives us:
\[
z_k^t = \sqrt{\frac{2\ln(NKt)}{n_k^t}} \le \sqrt{ \frac{2\ln(NKT)}{m_\tau/2}} = \sqrt{\frac{4\ln(NKT)}{m_\tau}}.
\]

Because $s^{-(1-\gamma)} \leq 1$ for every $s \geq 1$, we know that:

\[
\begin{aligned}
Km_\tau
&=
K+\sum_{s=K+1}^{\tau-1}s^{-(1-\gamma)}
\\
&\geq
\sum_{s=1}^{\tau-1}s^{-(1-\gamma)}
\\
&\geq
\int_1^\tau x^{-(1-\gamma)}\,dx
\\
&=
\frac{\tau^\gamma-1}{\gamma}.
\end{aligned}
\]

Therefore, $m_\tau \geq \frac{\tau^\gamma-1}{\gamma K}$, so:

\[
\begin{aligned}
z_k^t
&\le
\sqrt{
\frac{4\ln(NKT)}
{(\tau^\gamma-1)/(\gamma K)}
}
\\
&=
\sqrt{
\frac{4\gamma K\ln(NKT)}
{\tau^\gamma-1}
}.
\end{aligned}
\]

This proves the third and final claim.

\end{proof}

%% file: proofs/exploitation-phase-bound.tex
Finally, we use the shrinking confidence radii to show that, once they become small relative to the minimum reward gap, every agent’s true favourite set is recovered and exploitation rounds no longer incur mismatches.

\begin{lemma}
\label{lm:mismatchbound}
Define a mismatch to refer to a step where the actual favourite arm set does not equal the estimated favourite arm set for at least one agent. With high probability, in the exploitation phase, the number of mismatches is bounded by $O\left(\left[\frac{(1+\alpha)^2\gamma K \ln{(NKT)}}{\Delta_{\min}^2}\right]^{\frac{1}{\gamma}}\right)$.
\end{lemma}
\begin{proof}
    Let $F_i$ represent agent $i$'s actual favourite arms set, and let $j^* \in F_i$ and $k \notin F_i$, and $\Delta = \Delta_{i, j^*, k} > 0$, where $\Delta_{i, j^*, k} = \mu^*_{i,j^*} - \mu^*_{i,k}$ (in the case where every arm is a favourite arm for agent $i$, the problem is trivial for that agent. So we are only considering cases where agent $i$ has a non-favourite arm). Let $j \in \arg \max_j \hat{\mu}^t_{i,j}$ be an estimated favourite arm for agent $i$. To bound the regret for the exploitation phase, we must prove two things. First, that no true favourite arm will ever be excluded from the estimated favourite set, and second, we must bound the number of steps until a non-favourite arm is excluded from the set.

    Take some true favourite arm $j^* \in F_i$, and let $\mu^*_i$ be the true value of $i$'s favourite arms, or $\mu^*_i = \max_\ell \mu^*_{i,\ell}$. Using Corollary \ref{cor:meanconc-uniform}, we know that $\hat{\mu}^t_{i,j^*} \geq \mu^*_i - z^t_{j^*}$ with high probability for all $t \in \{K+1,...,T\}$. Since $\alpha \geq 1$, we have that $\hat\mu_{i,j^*}^t+\alpha z_{j^*}^t \geq \mu_i^*+(\alpha-1)z_{j^*}^t \geq \mu_i^*$. Then we have the empirically best arm, $j \in \arg \max_j \hat{\mu}^t_{i,j}$, where on the clean event with high probability we have $\hat{\mu}^t_{i,j} \leq \mu^*_{i,j} + z^t_j$ for all $t \in \{K+1,...,T\}$, which means $\hat\mu_{i,j}^t-\alpha z_j^t \leq \mu_{i,j}^*-(\alpha-1)z_j^t \leq \mu_i^*$. Therefore, we have that:

    \[
    \hat\mu_{i,j^*}^t+\alpha z_{j^*}^t \geq \mu^*_i \geq  \hat\mu_{i,j}^t-\alpha z_j^t
    \]

    This is the exact condition for $j^*$ to remain in the estimated favourite set. Therefore, $F_i$ is a subset of the estimated favourite set for all time steps $t \in \{K+1,...,T\}$ with high probability.

    Now we must bound the number of time steps it takes for a non-favourite arm to be excluded from the estimated favourite arm set. In order for an arm $k$ to remain in the favourite arm set, we must satisfy the following condition: $\hat{\mu}^t_{i,k} + \alpha z^t_k \geq \hat{\mu}^t_{i,j} - \alpha z^t_j$ (note that $j^*$ may or may not be the same arm as $j$). Thus, we want to find at what time step $t$, the following will be true: $\hat{\mu}^t_{i,k} + \alpha z^t_k < \hat{\mu}^t_{i,j} - \alpha z^t_j$.

    We know, by definition, that $\hat{\mu}^t_{i,j} \geq \hat{\mu}^t_{i,j^*}$, as $j$ represents the arm that has the empirically highest mean at time $t$, so the estimate is at least as high as $j^*$'s estimate. This provides a lower bound to the right-hand side of our equation. We can then use the more conservative exclusion inequality $\hat{\mu}^t_{i,k} + \alpha z^t_k < \hat{\mu}^t_{i,j^*} - \alpha \max\{ z_j^t , z_{j^*}^t \} $. We can replace our other inequality with more conservative inequalities, as we are simply looking for an upper bound. If we replace our inequality with a more conservative inequality, then we know that if the conservative inequality is satisfied, then the original inequality is satisfied.

    From Corollary \ref{cor:meanconc-uniform}, we know with high probability that $\hat{\mu}^t_{i,k} \leq \mu^*_{i,k} + z^t_k$, providing an upper bound to the left-hand side of our inequality. We can then use the new, more conservative exclusion inequality $\mu^*_{i,k} + z^t_k + \alpha z^t_k < \hat{\mu}^t_{i,j^*} - \alpha \max\{ z_j^t , z_{j^*}^t \} $. By the same corollary, we also know that $\hat{\mu}^t_{i,j^*} \geq \mu^*_{i,j^*} - z^t_{j^*} \geq \mu^*_{i, j^*} - \max\{ z_j^t , z_{j^*}^t \} $. Similarly, we have a new lower bound for our right-hand side, so we can replace our exclusion inequality with $\mu^*_{i,k} + z^t_k + \alpha z^t_k < \mu^*_{i,j^*} - \max\{ z_j^t , z_{j^*}^t \} - \alpha \max\{ z_j^t , z_{j^*}^t \} $. Rearranging and factoring gives: $\Delta_{i,j^*,k} > (1+\alpha)(z^t_k + \max\{ z_j^t , z_{j^*}^t \})$.

    For our bound, we are primarily concerned with the smallest $\Delta$, so we use
    \[
    \Delta_{\min}
    :=
    \min_{i:\, |F_i| < K}
    \min_{j \in F_i}
    \min_{k \notin F_i}
    (\mu_{i,j}^* - \mu^*_{i,k}) > 0.
    \]
    Our inequality becomes
    \[
    \Delta_{\min}
    >
    (1+\alpha)\bigl(z^t_k + \max\{z_j^t,z_{j^*}^t\}\bigr).
    \]

    For a time step $K+1 \leq \tau \leq T$, recall from Lemma \ref{lm:ucbdecay} that
    \[
    m_\tau
    =
    1+\frac{1}{K}\sum_{s=K+1}^{\tau-1}s^{\gamma-1},
    \qquad
    r_\tau
    =
    \sqrt{\frac{4\ln(NKT)}{m_\tau}}.
    \]
    From Lemma \ref{lm:ucbdecay}, for every arm $\ell \in [K]$ and every time step $t \in \{\tau,\dots,T\}$, we have $z^t_\ell \leq r_\tau$. This applies in particular to $k$, $j$, and $j^*$. Hence
    \[
    z^t_k + \max\{z^t_j,z^t_{j^*}\}
    \leq 2r_\tau,
    \]
    so it is sufficient to ensure that
    \[
    2(1+\alpha)r_\tau < \Delta_{\min}.
    \]

    Substituting $r_\tau$ with its definition, we have $2(1+\alpha) \sqrt{4\ln(NKT)/m_\tau} < \Delta_{\min}$. Solving for $m_\tau$, we get:

    \[
    m_\tau > \frac{16 (1+\alpha)^2 \ln(NKT)}{\Delta_{\min}^2}
    \]

    Recall from Lemma \ref{lm:ucbdecay}, we have that $m_\tau \geq \frac{\tau^\gamma - 1}{\gamma K}$. So our inequality is guaranteed to hold whenever

    \[
    \frac{\tau^\gamma-1}{\gamma K} > \frac{16(1+\alpha)^2\ln(NKT)}{\Delta_{\min}^2}.
    \]

    Rearranging, we get:

    \[\tau
    >
    \left[
    1+
    \frac{16(1+\alpha)^2\gamma K\ln(NKT)}
    {\Delta_{\min}^2}
    \right]^{1/\gamma}
    \]

    Let $\tau_0$ be the smallest integer satisfying this inequality and $\tau_0 \geq K+1$. If $\tau_0 \leq T$, then on the intersection of the high probability arm-count event from Lemma \ref{lm:ucbdecay} and the high probability uniform concentration event from Corollary \ref{cor:meanconc-uniform}, every true favourite set is recovered for every agent for every time step from $\tau_0$ onwards, until the end of the horizon $T$. In other words, $\hat{F}_i^t = F_i$ for all agents $i \in [N]$ and for every $t \in \{\tau_0,...,T\}$. Since exploitation phase mismatches can only occur during time steps $t \in \{K+1,...,\tau_0-1\}$, their number is at most $\max\{0,\tau_0-K-1\}$. If $\tau_0 > T$, the number of exploitation phase mismatches is trivially at most $T-K < \tau_0$. Therefore, the final regret bound for the exploitation phase can be written as:

    \[
    O\left(\left[\frac{(1+\alpha)^2\gamma K \ln{(NKT)}}{\Delta_{\min}^2}\right]^{\frac{1}{\gamma}}\right)
    \]

\end{proof}

%% file: proofs/full-pf-regret-bound.tex
\begin{proof}[Proof of Theorem~\ref{thm:pf_regret}]
    Let $F_i$ represent agent $i$'s true favourite-arm set under the true means $\mu^*$. At each adaptive round $t \in \{K+1,...,T\}$, the algorithm uses the estimates $\hat{\mu}^t$ and $z^t_k = \sqrt{\frac{2\ln{(NKt)}}{n^t_k}}$ for arm $k$, where $n^t_k$ represents the number of times arm $k$ has been pulled by time $t$.

    Additionally, for $j \in F_i$ and $k \notin F_i$, define the gap $\Delta_{i,j,k} := \mu^*_{i,j} - \mu^*_{i,k} > 0$, and let $\Delta_{\min} = \min_{i:\, |F_i| < K}\ \min_{j \in F_i}\ \min_{k \notin F_i} \Delta_{i,j,k} > 0$. Note $\Delta_{\min}$ is defined as $\infty$ if all agents are indifferent between all arms, that is, $|F_i| = K$ for all agents $i$, in which case there are no non-favourite arms to exclude, so we only need to consider situations where there exists some non-favourite arm for at least one agent. We also have $\hat{F}_i^t = \{k : \hat{\mu}^t_{i,k} + \alpha z^t_k \geq \hat{\mu}^t_{i,j} - \alpha z^t_j\}$, where $j = \arg \max_j \hat{\mu}^t_{i,j}$ for some agent $i$.

    In our setting, we define \emph{regret} to refer to the initialization cost and the number of subsequent mismatches. Specifically, we charge each of the first $K$ forced initialization rounds one unit of regret, and we count a \emph{mismatch} at time $t \in \{K+1,...,T\}$ if $\exists i : \hat{F}_i^t \neq F_i$. Let $M_t = 1$ for $t \in \{1,...,K\}$, and let $M_t = \mathbf{1}\{\exists i : \hat{F}_i^t \neq F_i\}$ for $t \in \{K+1,...,T\}$. Further, at each adaptive round $t \in \{K+1,...,T\}$, with probability $p^t_{\mathrm{rand}} = t^{-(1-\gamma)}$ the algorithm pulls an arm at random, which we call the exploration phase. Let $E_t \sim \text{Bernoulli}(t^{-(1-\gamma)})$ indicate whether or not $t$ was a round in which random exploration occurred. Thus, our regret bound can be defined as:

    \[
    \begin{aligned}
    R^{\PF}(T)
    &= \sum_{t=1}^T M_t
    = K+\sum_{t=K+1}^T M_t
    \\
    &\leq K+\underbrace{\sum_{t=K+1}^T E_t}_{(a) \; \text{exploration}}
    +\underbrace{\sum_{t=K+1}^T(1-E_t)M_t}_{(b) \; \text{exploitation}}.
    \end{aligned}
    \]

    From Lemma \ref{lm:randbound}, we have that part (a) has the bound $O(\frac{T^\gamma}{\gamma})$, and from Lemma \ref{lm:mismatchbound}, we have that part (b) has bound $O([\frac{(1+\alpha)^2\gamma K \ln{(NKT)}}{\Delta_{\min}^2}]^{\frac{1}{\gamma}})$. Taking a union bound over the high probability events in these lemmas, parts (a) and (b) hold simultaneously with high probability. Using part (a) and (b) together, we have our final regret bound:

    \[
    R^{\PF}(T) = O\left(K+\frac{T^\gamma}{\gamma} +\left(\frac{(1+\alpha)^2\gamma K \ln{(NKT)}}{\Delta_{\min}^2}\right)^{\frac{1}{\gamma}}\right)
    \]
\end{proof}

%% file: sections/comparison-framework.tex
\section{Outcome-Based Comparison Framework}

While procedural fairness evaluates influence distribution, outcome-based objectives evaluate the resulting utility vector across two key dimensions: total utility and distribution. We isolate these dimensions using utilitarian welfare and inequality minimization as endpoints. Additionally, we benchmark against egalitarian welfare, Nash welfare, and generalized Gini welfare as established fairness objectives.

\subsection{Diagnostic Endpoints}


The tradeoff between utilitarianism and inequality is a foundational debate in economics \citep{okun, ATKINSON1970244, inequalityaversion}, political science \citep{Kenworthy_1995, ditributionandpolitics}, and fair division \citep{freeman2019equitableallocationsindivisiblegoods, Caragiannis2012}. Utilitarian welfare maximizes total expected utility, whereas inequality minimization reduces disparities between agents' expected utilities. These endpoints allow us to measure what procedural representation costs---or preserves---on each outcome dimension.

\subsubsection{Inequality Minimization}

Inequality minimization represents equal outcomes, where the policy aims to give each agent as close to equal expected rewards as possible, which is an objective inspired by equitability in the fair division literature \citep{freeman2019equitableallocationsindivisiblegoods, Caragiannis2012}. This principle is most useful in contexts where balance across agents matters more than efficiency or giving each agent equal voice. We define inequality minimization as follows:

\begin{definition}[Inequality Minimization]
    A policy \( P = (p_1, \dots, p_K) \) minimizes inequality if it minimizes differences in expected rewards among agents. Formally, the inequality-minimizing policy is
    \[
    P_{\IM}^*
    \in
    \arg\min_{Q \in \Delta^K}
    \frac{2}{N(N-1)}
    \sum_{i>j}
    \left(
    \sum_{k=1}^{K}
    (Q_k\mu^*_{i,k}-Q_k\mu^*_{j,k})
    \right)^2.
    \]
\end{definition}

We also define a way to measure inequality minimization, or the equality score:
\[
\IM(\mu^*,P)=1-\left|D(P)-D(P_{\IM}^*)\right|,
\]
where $P_{\IM}^*$ is an inequality-minimizing policy and
\[
D(P)
=
\frac{2}{N(N-1)}
\sum_{i>j}
\left(
\sum_k
(p_k\mu^*_{i,k}-p_k\mu^*_{j,k})
\right)^2.
\]
Please note that $D(\cdot) \in [0, 1]$, so $\IM(\mu^*, P) \in [0, 1]$.

\subsubsection{Utilitarianism}

Utilitarianism is the notion of maximizing the overall utility of the group. This principle is most appropriate in efficiency-driven domains where aggregate outcomes matter most.

\begin{definition}[Utilitarianism]
A policy \( P = (p_1, \dots, p_K) \) is utilitarian if it maximizes the sum of expected utilities across agents. Formally, the utilitarian policy is
\[
P_{\UT}^*
\in
\arg\max_{Q \in \Delta^K}
\sum_{i=1}^{N}\sum_{k=1}^{K}Q_k\mu^*_{i,k}.
\]
\end{definition}

The score for utilitarianism is
\[
\UT(\mu^*,P)
=
\frac{
\sum_{i\in[N]}\sum_{k\in[K]}p_k\mu^*_{i,k}
}{
\sum_{i\in[N]}\sum_{k\in[K]}(P_{\UT}^*)_k\mu^*_{i,k}
}.
\]
This provides a percentage share of what the policy is achieving with respect to what can be achieved. If the denominator is zero, then $\UT(\mu^*, P)=1$.

\input{proofs/examples-fair-policies}

\subsection{Secondary Benchmarks}

We additionally compare against three established objectives for selecting a policy: egalitarian welfare \citep{moulin}, Nash social welfare (from the NashUCB algorithm \citep{hossain2021fair}), and the generalized Gini index \citep{busa2017multi}. Their respective policy objectives are
\[
\begin{aligned}
P_{\EG}^*
&\in \arg\max_{Q\in\Delta^K}\min_{i\in[N]}U_i(Q),\\
P_{\NSW}^*
&\in \arg\max_{Q\in\Delta^K}\prod_{i=1}^{N}U_i(Q),\\
P_{\GGI}^*
&\in \arg\max_{Q\in\Delta^K}\sum_{i=1}^{N}w_iU_{(i)}(Q),
\end{aligned}
\]
where $U_{(1)}(Q)\leq\cdots\leq U_{(N)}(Q)$ and $w_1\geq\cdots\geq w_N\geq 0$. In this paper, we use $w_i = 2^{-(i-1)}$ for $i=1,\dots,N$. We include these objectives as representative approaches to fairness to compare with procedural fairness; however, we do not claim this to be an exhaustive taxonomy.

%% file: proofs/examples-fair-policies.tex
\subsubsection{Illustrative Comparison}
\label{sec:fair-policy-example}
Consider a multi-agent multi-armed bandit setting with $N = 3$ and $K = 2$, and the following reward structure:
\[\mu = \begin{bmatrix}1 & 0 \\ 1 & 0 \\ 0 & 1\end{bmatrix}\]

In this scenario, for procedural fairness, agents 1 and 2 would place their $\frac{1}{N}$ probability mass on the first arm, and agent 3 would place their probability mass on the second arm. This results in a policy of $(\frac{2}{3}, \frac{1}{3})$. An \textit{inequality-minimizing} policy would be $P = (\frac{1}{2}, \frac{1}{2})$, as each agent's expected utility from such a policy would be 0.5. A \textit{utilitarian} policy would be $P = (1, 0)$, as the leftmost arm has an overall utility among all agents of 2, whereas the rightmost arm has an overall utility of 1 among all agents.

%% file: sections/theoretical-results.tex
\section{Theoretical Results}

\subsection{Procedural Fairness and the Core}

The core is a stability notion originating in cooperative game theory~\cite{shapley1971cores}. In the context of public decision-making~\cite{fain2018fair}, it represents a distribution over alternatives (arms) that no coalition of agents $A\subseteq [N]$ would have an incentive to deviate from, given their proportional share of probability (\(|A|/N\)). We define the core in our setting in two ways, considering both utility (outcome core) and decision share (procedural core).

\begin{definition}[Outcome Core]
    A distribution $P \in \Delta^K$ is in the outcome core if there is no coalition of agents $A \subseteq [N]$ and distribution $Q \in \Delta^K$ such that $\frac{|A|}{N}U_i(Q) \geq U_i(P)$ for all $i \in A$, with at least one strict inequality.
\end{definition}

The procedural core adapts the classic notion of the core in cooperative game theory to the setting of procedural fairness. Rather than considering the agents' expected utility, we consider their decision share: the total probability mass assigned to the agents' most preferred arms. This reflects the purpose of procedural fairness, that representation is valuable independently of the reward produced by some policy. As a result, the procedural core means that no coalition can use its proportional share of policy mass to improve every member's representation, with at least one strict improvement.

\begin{definition}[Procedural Core]
     Recall that $\mu^*$ is the true reward-mean matrix. Let $F_i = \{k \in [K] \mid \mu^*_{i,k} = \max_{j\in [K]} \mu^*_{i, j}\}$ denote agent $i$'s favourite arms. Define a binary vector $X_i \in \{0, 1\}^K$ for each agent $i$, where $X_i[k]$ is 1 if $k \in F_i$ and 0 otherwise. Thus, given a policy $P$, the decision share of agent $i$ is defined as $\beta_i(P) = \sum_{k=1}^K X_i[k] p_k = P \cdot X_i$. As with the outcome core, a policy $P$ is in the core if there is no coalition of agents $A \subseteq [N]$ and distribution $Q \in \Delta^K$ such that $\frac{|A|}{N}\beta_i(Q) \geq \beta_i(P) \quad \forall i \in A$, with at least one strict inequality.
\end{definition}

The procedural core carries interesting implications for our setting, namely, that a utility-based Nash welfare-maximizing distribution does not necessarily lie in the procedural core.

\begin{theorem}
\label{thm:OutcomeNashNotInProceduralCore}
A utility-based Nash welfare-maximizing distribution need not lie in the procedural core.
\end{theorem}

\input{proofs/ubnw-not-in-core}

Mirroring the known connection between Nash welfare and the core \citep{fain2016coreparticipatorybudgetingproblem}, we show this carries to our setting.

\begin{theorem}
\label{thm:RRimpliesPC}
A policy selected by maximizing the decision-share-based Nash welfare objective, 

\[
\prod_{i=1}^N \sum_{k \in F_i} p_k
\]

is in the procedural core.
\end{theorem}

\input{proofs/dbnw-implies-core}

We also show that procedural fairness requires an exact structural characterization as coalitional proportionality:

\begin{lemma}
    \label{lm:pfcoalprop}
    Let $S \subseteq [N]$, $U_S = \bigcup_{i \in S} F_i$, and $P(U_S) = \sum_{k \in U_S} p_k$. Then, $P$ is a procedurally fair policy if and only if $P(U_S) \geq \frac{|S|}{N}$ for every $S \subseteq [N]$.
\end{lemma}

\input{proofs/coalitional-proportionality}

Using this characterization, we can then show that membership in the procedural core implies procedural fairness:

\begin{theorem}
\label{thm:PCimpliesPF}
Procedural core implies procedural fairness.
\end{theorem}

\input{proofs/core-implies-pf}

\subsection{Impossibility Results}

We show that different notions of fairness provably conflict, and that there are instances where one must be prioritized over the others. These results make explicit that no single policy can be guaranteed to satisfy all objectives simultaneously, underscoring that procedural fairness is not just another outcome-based criterion but a distinct axis of fairness.

\begin{proposition}
\label{thm:impossibility}
There exist reward structures for which no policy can simultaneously achieve perfect procedural fairness and perfect inequality minimization.
\end{proposition}

\input{proofs/impossibility-inequality}

\begin{proposition}
\label{thm:impossibility2}
There exist reward structures for which no policy can simultaneously achieve perfect procedural fairness and perfect utilitarianism.
\end{proposition}

\input{proofs/impossibility-utilitarianism}

\subsection{Pareto Incomparability}
We now compare utility-based Nash social welfare, which favours allocations where all agents receive non-trivial utility, to procedural fairness in terms of Pareto dominance along our metrics. While utility-based Nash social welfare doesn't account for procedural fairness, it serves as a strong comparator due to its prominence in the literature. Formally:

\begin{definition}[Pareto Dominance]
A policy \(P'\in\Delta^K\) \emph{Pareto dominates} another policy \(P\in\Delta^K\) with respect to the fairness metrics \((\PF,\,\IM,\,\UT)\) if $\PF(\mu^*, P') \;\ge\; \PF(\mu^*, P)$, $\IM(\mu^*, P') \geq \IM(\mu^*, P)$, and $\UT(\mu^*, P') \;\ge\; \UT(\mu^*, P)$, and at least one of these inequalities is strict.
\end{definition}

\begin{theorem}
\label{thm:nw_pf_incomparable}
Utility-based Nash welfare-optimal policies and procedurally fair policies are not guaranteed to Pareto dominate one another.
\end{theorem}

\input{proofs/pareto-incomparability}

This result underscores the critical limitation of relying on utility-based Nash welfare as a default fairness benchmark. While Nash welfare is widely used because it balances efficiency and inequality, it imposes an outcome-centric criterion that does not preserve equal representation in the decision-making process. Treating Nash welfare as sufficient effectively prioritizes efficiency over equal voice, a normative choice that is often left implicit.

%% file: proofs/ubnw-not-in-core.tex
\begin{proof}
We prove this through a counterexample. Consider a setting with two agents ($N=2$) and two arms ($K=2$). The reward matrix
\[
\mu \;=\;
\begin{pmatrix}
1 & 0.99 \\[6pt]
0 & 1
\end{pmatrix}
\]

Agent 1's favourite arm is arm~1, and Agent 2's favourite arm is arm~2. We will now find the distribution over arms that maximizes the utility-based Nash welfare.

Let $p_2$ be the probability of pulling arm~2 and $p_1 = 1 - p_2$ be the probability of pulling arm~1. Then:
\[
\begin{aligned}
& \text{Agent~1's expected reward: }   & U_1(p) &= 1 \cdot p_1 + 0.99 \cdot p_2 = p_1 + 0.99\,p_2, \\[4pt]
& \text{Agent~2's expected reward: }   & U_2(p) &= 0 \cdot p_1 + 1     \cdot p_2 = p_2.
\end{aligned}
\]
Utility-based Nash welfare (the product of expected rewards) is
\[
\Bigl(p_1 + 0.99\,p_2\Bigr)\;\cdot\;p_2
=\;
\bigl(p_1\bigr)\,p_2 + 0.99\,p_2^2.
\]
Since $p_1 = 1 - p_2$, this equals
\[
(1-p_2)p_2 + 0.99\,p_2^2
=
p_2 - p_2^2 + 0.99\,p_2^2
=
p_2 - 0.01\,p_2^2.
\]
Taking the derivative:
\[
\frac{d}{dp_2} \bigl(p_2 - 0.01\,p_2^2\bigr)
=\; 1 \;-\; 0.02\,p_2.
\]
On the interval $p_2 \in [0,1]$, this derivative never vanishes (it is $1 - 0.02\,p_2 > 0$ for all $p_2 \in [0,1]$). Hence the function is strictly increasing over $[0,1]$, with its maximum at the boundary $p_2 = 1$.

Thus the unique maximizer of utility-based Nash welfare is
\[
(p_1, p_2) \;=\; (0,\,1).
\]

We can now show that this distribution is not in the procedural core. Under \((0,1)\), the \emph{decision share} of an agent is the total probability on that agent's favourite arm(s). Thus:
\[
\beta_1\bigl((0,1)\bigr) = p_1 = 0
\]
\[
\beta_2\bigl((0,1)\bigr)
= p_2 = 1
\]
We can look at the single-agent coalition $C = \{\text{Agent }1\}$. By deviating to the distribution $(1,0)$ (which puts probability~1 on arm~1), Agent~1's procedural utility becomes
\[
\beta_1\bigl((1,0)\bigr) = 1.
\]
Since $\tfrac{|C|}{N} = \tfrac{1}{2}$, we scale this by $\tfrac12$ to obtain
\[
\tfrac{|C|}{N}\cdot \beta_1\bigl((1,0)\bigr)
=
\tfrac12 \times 1
=
0.5
\quad>\quad
0
=
\beta_1\bigl((0,1)\bigr).
\]
Hence Agent~1 alone can \emph{strictly} increase their procedural utility when switching from $(0,1)$ to $(1,0)$. By definition of the procedural core, \((0,1)\) is therefore \emph{blocked} and cannot be in the procedural core. Thus, the unique distribution maximizing utility-based Nash welfare in this instance, \((0,1)\), is not in the procedural core because a single-agent coalition has a profitable deviation. This shows that the utility-based Nash welfare maximizing solution is not guaranteed to be in the procedural core.
\end{proof}

%% file: proofs/dbnw-implies-core.tex
\begin{proof}
Let $P^*$ be the selected policy and define each agent's decision share as $\beta_i(P) = \sum_{k \in F_i} p_k$. Assume, for sake of contradiction, that coalition $C$ blocks $P^*$ through some policy $Q$. Let $\alpha = |C|/N$.

By the definition of blocking, $\alpha \beta_i(Q) \geq \beta_i(P^*)$ for every $i \in C$ with at least one strict inequality. Equivalently,

\[
\frac{\beta_i(Q)}{\beta_i(P^*)} > \frac{1}{\alpha} = \frac{N}{|C|}
\]

for at least one $i \in C$. Note that $\beta_i(P^*)$ cannot be 0, because it is the policy selected to maximize the decision-share-based Nash welfare. Therefore,

\[
\sum_{i\in C}
\frac{\beta_i(Q)}{\beta_i(P^*)}
>
|C|\frac{N}{|C|}
=
N
\]

Now consider moving a small amount, $\epsilon$, from $P^*$ towards $Q$. In other words: $P_\epsilon = (1-\epsilon) P^* + \epsilon Q$. Because each $\beta_i$ is linear, we have that: $\beta_i(P_\varepsilon) = (1-\varepsilon)\beta_i(P^*) + \varepsilon\beta_i(Q).$

So we can take the derivative of the objective function, $\prod_{i=1}^N \sum_{k \in F_i} p_k = \prod_{i=1}^N \beta_i(P)$ w.r.t $\epsilon$ at $\epsilon = 0$. If the derivative is positive, that means that there exists some small $\epsilon$ such that the value of the objective function can increase. A policy that maximizes the objective function should have a derivative w.r.t $\epsilon$ at $\epsilon = 0$ equal to 0 (or less than 0 if negative $\epsilon$ would make it an infeasible solution). Note that maximizing $\prod_{i=1}^N \beta_i(P)$ is equivalent to maximizing $\sum_{i=1}^N \log \beta_i(P)$, so we use this in place of the product. Therefore,
\[
\frac{d}{d\varepsilon}
\sum_{i=1}^{N}\log\beta_i(P_\varepsilon)
=
\sum_{i=1}^{N}
\frac{\beta_i(Q)-\beta_i(P^*)}
{(1-\varepsilon)\beta_i(P^*)+\varepsilon\beta_i(Q)}.
\]
Plugging in $\varepsilon=0$, we obtain
\[
\sum_{i=1}^{N}
\frac{\beta_i(Q)-\beta_i(P^*)}{\beta_i(P^*)}
=
\sum_{i=1}^{N}
\left(
\frac{\beta_i(Q)}{\beta_i(P^*)}-1
\right)
=
\left( \sum_{i=1}^{N}
\frac{\beta_i(Q)}{\beta_i(P^*)} \right)
-
N.
\]

From earlier, we have $\sum_{i\in C} \frac{\beta_i(Q)}{\beta_i(P^*)} > N.$ Since every term in the sum is non-negative (no $\beta_i$ can be less than 0 under any policy, as $p_k$ must always be at least 0 under any policy for any arm), it follows that $\sum_{i=1}^{N} \frac{\beta_i(Q)}{\beta_i(P^*)} > N.$ This means that the derivative at $\varepsilon=0$ is strictly positive, and for a sufficiently small positive $\varepsilon$, $P_\varepsilon$ has strictly greater decision-share-based Nash welfare than $P^*$. This
contradicts the global optimality of $P^*$. Hence, no coalition blocks $P^*$,
and $P^*$ lies in the procedural core.

\end{proof}

%% file: proofs/coalitional-proportionality.tex
\begin{proof}
    We first prove the forward direction: if $P$ is procedurally fair, then $P(U_S) \geq |S|/N$ for every $S \subseteq [N]$.

    Suppose $P$ is procedurally fair. Then there exists a non-negative allocation $y_{i,k}$ that satisfies $\sum_{k=1}^K y_{i,k} = 1/N$ for every agent $i$, and $y_{i,k} = 0$ when $k \notin F_i$, and $\sum_{i=1}^N y_{i,k} = p_k$ for every arm $k$. Consider any coalition $S \subseteq [N]$. The members of $S$ collectively contribute $\sum_{i \in S} \sum_{k=1}^K y_{i,k} = \sum_{i \in S} 1/N = |S|/N$.

    Each agent $i \in S$ can only allocate probability to arms in $F_i$ (since $P$ is procedurally fair), so the probability contributed by coalition $S$ to arms in $U_S$ is therefore $\sum_{i \in S} \sum_{k \in U_S} y_{i,k} = |S|/N$. Furthermore, the coalition's contribution to $U_S$ cannot exceed the total contribution of all agents to $U_S$ (an agent may have an arm in $U_S$ that is its favourite, but it is not in $S$). Therefore,

    \[
    \sum_{i \in S} \sum_{k \in U_S} y_{i,k} \leq \sum_{i = 1}^N \sum_{k \in U_S} y_{i,k}
    \]

    Furthermore, we also require through procedural fairness that $\sum_{i=1}^N y_{i,k} = p_k$. Therefore,

    \[
    \sum_{i=1}^N \sum_{k \in U_S} y_{i,k} = \sum_{k \in U_S} p_k = P(U_S) \geq \sum_{i \in S} \sum_{k \in U_S} y_{i,k} = \frac{|S|}{N}
    \]

    Meaning $P(U_S) \geq |S|/N$, so the forward condition is proven true.

    Now we prove the reverse. Assume $P(U_S) \geq |S|/N$ for every coalition $S \subseteq [N]$. We must now show that this means that $P$ is procedurally fair. We do this using a flow network.

    Create a directed network with a source node $s$, one node for every agent $i$, one node for every arm $k$, and a sink node $t$. We add an edge from $s$ to each agent $i$ with capacity $1/N$, an edge from each agent $i$ to arm $k \in F_i$ with capacity 1, and an edge from each arm $k$ to $t$ with capacity $p_k$. We will prove that this network admits a flow of value 1. By the max-flow/min-cut theorem, it is enough to prove that every $s-t$ cut has capacity at least 1.

    Consider an arbitrary cut. Let $A$ be the agents on the source side of the cut, and $B$ be the arms on the source side of the cut. We consider two cases. In the first case, an edge passing from an agent to an arm crosses the cut. Every edge from an agent to an arm has capacity 1 from construction, so such a cut has capacity at least 1.

    In the second case, an edge passing from an agent to an arm does not cross the cut. In this scenario, if agent $i$ is on the source side, then every arm in $F_i$ must also be on the source side. Otherwise that edge would cross the cut. Therefore, $U_A = \bigcup_{i \in A} F_i \subseteq B$. Therefore, the source-to-agent edges crossing the cut are precisely those incident to agents outside $A$. The total capacity of these edges is $\frac{N - |A|}{N}$ (since the number of agents not in $A$ is $N - |A|$ and each of these edges has capacity of $1/N$ by construction). Moreover, the edges going from an arm to the sink that cross the cut are edges that go from an arm in $B$ to the sink. The capacity of each of these edges is $p_k$, so we define the total capacity as $P(B) = \sum_{k \in B} p_k$. Therefore, the total cut capacity is defined as $\frac{N-|A|}{N} + P(B)$. Because $U_A \subseteq B$, we know that $P(B) \geq P(U_A)$ (since all elements in the sum are non-negative).

    Since we assume that $P(U_S) \geq |S|/N$, we know that the cut capacity satisfies:

    \[
    \begin{aligned}
    \frac{N-|A|}{N}+P(B)
    &\geq
    \frac{N-|A|}{N}+P(U_A)\\
    &\geq
    \frac{N-|A|}{N}+\frac{|A|}{N}\\
    &=1.
    \end{aligned}
    \]

    Therefore, every $s-t$ cut has capacity of at least 1. Further, we know that there is a cut of capacity of exactly 1: place only the source $s$ on the source side. The $N$ edges that cross the cut are the edges going from the source to an agent, who each have a capacity of $1/N$, or a total capacity of 1. Therefore, the minimum cut has capacity 1, and by the max-flow/min-cut theorem, there exists a flow of value 1 in the network.

    Let $y_{i,k}$ denote the flow from agent $i$ to arm $k$. Because the network's flow is 1 and the total capacity leaving the source is also 1, every edge from $s$ to an agent must be saturated. Therefore, $\sum_{k=1}^K y_{i,k} = \frac{1}{N}$ for every agent $i$ in $[N]$. Each agent $i$ only has outgoing edges to arms in $F_i$, meaning $y_{i,k} = 0$ for any arm $k \notin F_i$. Total capacity entering the sink can be computed as $\sum_{k=1}^K p_k = 1$. Then, since the flow has value 1, every arm-to-sink edge must also be saturated, giving us $\sum_{i=1}^N y_{i,k} = p_k$ for every arm $k$. Therefore, $y$ is a procedurally fair decomposition of $P$, meaning that if $P(U_S) \geq |S|/N$, $P$ is a procedurally fair policy.

    Therefore, $P$ is a procedurally fair policy if and only if $P(U_S) \geq \frac{|S|}{N}$ for every $S \subseteq [N]$.
\end{proof}

%% file: proofs/core-implies-pf.tex
\begin{proof}
    Suppose for the sake of contradiction, that the policy $P$ lies in the procedural core but is not procedurally fair. By Lemma \ref{lm:pfcoalprop}, we know that since $P$ is not procedurally fair, there exists some nonempty coalition $S \subseteq [N]$ such that $P(U_S) < |S|/N$. Let $\alpha = |S|/N$. Therefore, $P(U_S) < \alpha$.

    Because $S$ is nonempty and every agent has at least one favourite arm (by definition), $U_S$ is nonempty. Choose any arm $k^* \in U_S$, and construct a second policy $Q = (q_1,..., q_K)$. We construct this policy by taking, for every arm $k \in U_S$, and scaling its probability by $1/\alpha$. Therefore, the total probability assigned to arms in $U_S$ immediately after scaling is $P(U_S)/\alpha$, which we know is strictly less than 1 since $P(U_S) < \alpha$. Then, take the remaining probability available and assign it to arm $k^*$.

    Formally, we define the policy $Q = (q_1,..., q_K)$ as the following, for each arm $k \in [K]$:

    \[
    q_k=
    \begin{cases}
    \dfrac{p_k}{\alpha},
    &
    k\in U_S,\ k\neq k^\star,\\[8pt]
    \dfrac{p_{k^\star}}{\alpha}
    +
    \left(1-\dfrac{P(U_S)}{\alpha}\right),
    &
    k=k^\star,\\[10pt]
    0,
    &
    k\notin U_S.
    \end{cases}
    \]

    Note that all probabilities in $Q$ are non-negative, and that
    \[
    \begin{aligned}
        \sum_{k=1}^K q_k
        &= \left(\sum_{k\in U_S \setminus \{k^*\}} \frac{p_k}{\alpha}\right) + \left(\frac{p_{k^*}}{\alpha} + 1 - \frac{P(U_S)}{\alpha}\right) + \left(\sum_{k \notin U_S} 0\right) \\
        &= \frac{P(U_S) - p_{k^*}}{\alpha} + \frac{p_{k^*}}{\alpha} + 1 - \frac{P(U_S)}{\alpha} \\
        &= 1
    \end{aligned}
    \]

    Therefore, $Q$ is a valid policy. Now, consider any agent $i \in S$. Since $F_i \subseteq U_S$:

    \[
    \begin{aligned}
        \beta_i(Q)
        &= \sum_{k \in F_i} q_k \\
        &= \frac{1}{\alpha} \sum_{k \in F_i} p_k + \left(1 - \frac{P(U_S)}{\alpha}\right) \mathbf{1}\{k^* \in F_i\} \\
        &= \frac{\beta_i(P)}{\alpha} + \left(1 - \frac{P(U_S)}{\alpha}\right) \mathbf{1}\{k^* \in F_i\}
    \end{aligned}
    \]

    Multiplying both sides by $\alpha$ gives us $\alpha \beta_i(Q) = \beta_i(P) + \left(\alpha - P(U_S)\right) \mathbf{1}\{k^* \in F_i\}$.

    Since we know that $P(U_S) < \alpha$, we have that $\alpha - P(U_S) > 0$. Therefore,

    \[
    \alpha \beta_i(Q) = \beta_i(P) + \left(\alpha - P(U_S)\right) \mathbf{1}\{k^* \in F_i\} \geq \beta_i(P)
    \]

    Or, more simply, $\alpha \beta_i(Q) \geq \beta_i(P)$. Further, since there exists at least one agent $i^* \in S$ such that $k^* \in F_{i^*}$, we have:

    \[
    \alpha \beta_{i^*}(Q) = \beta_{i^*}(P) + \alpha - P(U_S) > \beta_{i^*}(P)
    \]

    Therefore,

    \[
    \frac{|S|}{N} \beta_i(Q) \geq \beta_i(P) \qquad \forall i \in S
    \]

    with at least one strict inequality, meaning coalition $S$ blocks $P$ using policy $Q$, contradicting the assumption that $P$ lies in the procedural core. Therefore, every policy in the procedural core must be procedurally fair.
\end{proof}

%% file: proofs/impossibility-inequality.tex
\begin{proof}
Consider a situation where $N\geq 2$ and $K \geq 2$, where:

\[
\mu_1 = (M, 0,...,0), \quad \mu_i=(0,1,0, ..., 0) \quad \forall i \geq 2
\]

where $0 < M < 1$. Under procedural fairness, each agent must allocate $1/N$ of the probability to their preferred arm, yielding $P_{\PF} = (\frac1N, \frac{(N-1)}{N}, 0,..., 0)$.

On the other hand, we want to find which policies minimize inequality. Under some arbitrary policy $P = (p_1,...,p_K)$, the expected utility of agent $1$ is $Mp_1$, and the expected utility of every other agent is $p_2$. Therefore, the inequality objective is $\frac{2}{N(N-1)}\sum_{i>j} (U_i(P) - U_j(P))^2$. We can see that the only non-zero pairwise utility is between agent 1 and the other $N-1$ agents, so $D(P) = \frac{2}{N(N-1)} (N-1)(Mp_1 - p_2)^2 = 2(Mp_1 - p_2)^2/N$. We can calculate the value of $D$ for the procedurally fair policy $P_{\PF}$ as:

\[
D(P_{\PF}) = \frac{2}{N}\left(M\frac{1}{N} - \frac{N-1}{N}\right)^2 = \frac{2(N-1-M)^2}{N^3}
\]

We know that $N-1-M>0$, which means that $D(P_{\PF})>0$. However, we can see that the minimum value of $D$ is 0 ( $\bar{P}= (\frac{1}{M+1}, \frac{M}{M+1}, 0,..., 0)$), meaning that no policy can simultaneously be procedurally fair and inequality minimizing in this reward structure.
\end{proof}

%% file: proofs/impossibility-utilitarianism.tex
\begin{proof}
Consider a situation where $N\geq 2$ and $K \geq 2$, where:

\[
\mu_1 = (M, 0,...,0), \quad \mu_i=(0,1,0, ..., 0) \quad \forall i \geq 2
\]

where $0 < M < 1$. Under procedural fairness, each agent must allocate $1/N$ of the probability to their preferred arm, yielding $P_{\PF} = (\frac1N, \frac{(N-1)}{N}, 0,..., 0)$. On the other hand, the utility maximizing policy is $P_{\UT} = (0, 1, 0,...,0)$. Since agent 1 never gets their favourite arm under the utility maximizing policy, and any procedural policy is strictly worse than the utilitarian policy, we have a disparity and both goals cannot be maximized at the same time.
\end{proof}

%% file: proofs/pareto-incomparability.tex
\begin{proof}
Consider the following counterexample utility matrix:

\[
\mu =
\begin{bmatrix}
A & 1 \\
1 & B
\end{bmatrix}
\]

Here, $0 \leq A < B < 1$. Thus, we know for certain that the procedurally fair policy here $P_{\PF} = (\frac{1}{2}, \frac{1}{2})$. Further, we also know that the utilitarian policy would be $P_{\UT}=(0, 1)$, as $A < B$. Thus, the utility-based Nash welfare solution would be Pareto incompatible with the procedurally fair solution if the probability it places on the rightmost arm is strictly greater than 0.5 but strictly less than 1.

Let $p$ represent the probability we pull the leftmost arm, and $1 - p$ denote the probability we pull the rightmost arm. Let $P_{\NSW}=(p, 1-p)$. Then, for each agent, we have utilities:
\[
U_1(P_{\NSW}) = p(A-1)+1
\]

\[
U_2(P_{\NSW}) = B+p(1-B)
\]

Then, to find the optimal utility-based Nash welfare solution, we want to optimize:

\[
f(p) = U_1(P_{\NSW})U_2(P_{\NSW}) = (p(A-1)+1)(B+p(1-B))
\]
\[
f'(p) = B(A-1)+2p(1-B)(A-1)+1-B = 0
\]
\[
p = \frac{B(2-A)-1}{2(1-B)(A-1)}
\]

Let $A = 0.4$ and $B = 0.6$, which satisfies our constraints from earlier. Then, $p = \frac{1}{12}$ and $1-p=\frac{11}{12}>\frac{1}{2}$, but less than 1. We can therefore derive that:
\[
\PF(\mu, P_{\PF}) = 1 > \PF(\mu, P_{\NSW})
\]
and
\[
\UT(\mu, P_{\NSW}) > \UT(\mu, P_{\PF})
\]

Therefore, neither policy can guarantee Pareto dominance over the other.

\end{proof}

%% file: sections/experiments.tex
\section{Experiments}

\subsection{Synthetic Evaluation}

We evaluate the objectives over 7,776 synthetic instances varying the numbers of agents and arms, favourite-set sizes, and homogeneous versus heterogeneous favourite sets. Preferences are generated from uniform, single-peaked, impartial-culture, and Mallows models; rewards are assigned consistently with each agent's ranking. We compare procedural fairness, inequality minimization, utilitarianism, egalitarian welfare, Nash social welfare, and the Generalized Gini Index. For procedural fairness, we optimize decision-share-based Nash welfare.

\subsubsection{Synthetic Instance Generation}
\label{sec:synthetic-generation}

We conduct a full factorial sweep over the number of agents
$N\in\{2,\ldots,10\}$, the number of arms $K\in\{2,\ldots,10\}$, and a
maximum favourite-set size $f_{\max}\in\{1,\ldots,10\}$, excluding
combinations for which $f_{\max}>K$.

We consider two favourite-set conditions. In the homogeneous condition,
every agent has the same number of favourite arms:
$|F_i|=f_{\max}$ for every $i\in[N]$. In the heterogeneous condition, we
independently draw $f_i$ uniformly from
$\{1,\ldots,f_{\max}\}$ for each agent and set $|F_i|=f_i$.

Preference orderings are generated using PrefVoting
\citep{HollidayPacuit2025}. We consider the uniform, single-peaked
\citep{Conitzer_2009}, impartial-culture \citep{guilbaud1952theories}, and
Mallows \citep{mallows1957non} models. For the Mallows model, we use
$\phi\in\{0.01,0.25,0.5,0.75,0.99\}$.

For each agent, we independently sample $K$ reward means from a normal
distribution with mean $0.5$ and standard deviation $0.25$, truncated to
$[0,1]$. We assign the largest sampled value to the agent's highest-ranked
arm, the second-largest value to their second-ranked arm, and so on. When
an agent has multiple favourite arms, each favourite arm receives the
largest value; the remaining values are assigned to non-favourite arms in
descending order of preference. Thus, agents may associate different
reward means with the same arm.

The resulting sweep contains 7,776 valid instances. We use random seed 42
throughout. Computing the optimal policies for all objectives takes
approximately ten minutes on an Apple M2 Pro processor. These experiments
assume known reward means and therefore evaluate the objectives themselves,
rather than the online learning process.

\subsubsection{Synthetic Results}

\begin{figure}[t]
\centering
\includegraphics[width=\linewidth]{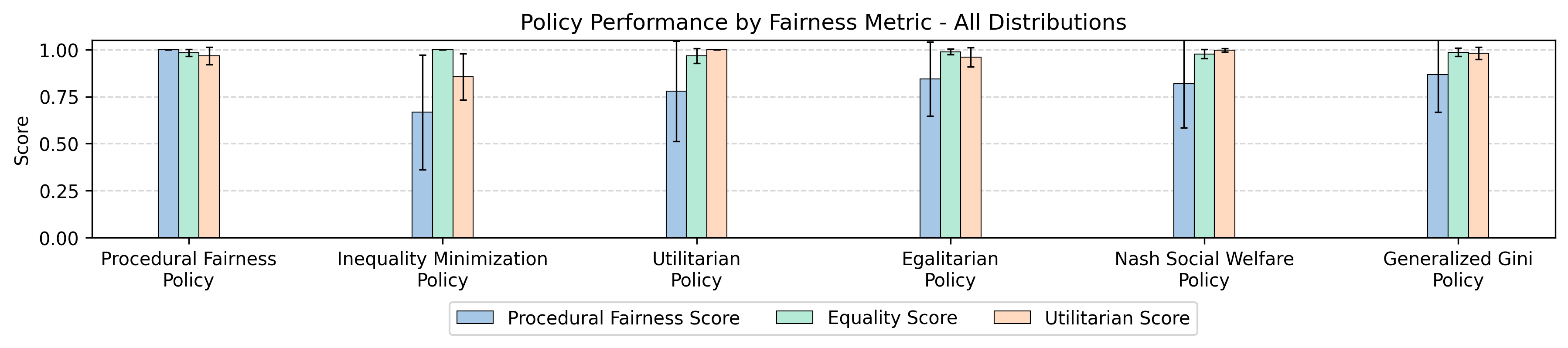}
\caption{Average fairness metrics per policy type. Each column refers to a specific policy maximizing a certain objective or fairness notion. Each bar represents a score for a fairness type, such as procedural fairness, inequality minimization, or utilitarianism, and the error bars represent one standard deviation from the mean across settings.}
\label{fig:all_distributions}
\end{figure}

\begin{table}[htbp]
\centering
\begin{tabular}{lrrr}
\toprule
 & PF Score & IM Score & UT Score \\
\midrule
PF Policy  & $1.00 \pm 0.00$ & $0.98 \pm 0.02$ & $0.97 \pm 0.05$ \\
\midrule
IM Policy  & $0.67 \pm 0.30$ & $1.00 \pm 0.00$ & $0.86 \pm 0.12$ \\
UT Policy  & $0.78 \pm 0.27$ & $0.97 \pm 0.04$ & $1.00 \pm 0.00$ \\
EG Policy  & $0.84 \pm 0.20$ & $0.99 \pm 0.02$ & $0.96 \pm 0.05$ \\
NSW Policy & $0.82 \pm 0.24$ & $0.98 \pm 0.02$ & $1.00 \pm 0.01$ \\
GGI Policy & $0.87 \pm 0.20$ & $0.99 \pm 0.02$ & $0.98 \pm 0.03$ \\
\bottomrule
\end{tabular}
\caption{Performance metrics for each policy. Reported as mean $\pm$ one standard deviation. Rows denote the policies optimizing each objective; columns denote fairness scores.}
\label{tab:performance_metrics}
\end{table}

Figure \ref{fig:all_distributions} shows the results of the different objectives in our experiment. As expected, our selected procedurally fair policies yield perfect procedural fairness scores (as the objective is designed to do). More importantly, however, they achieve the best balance across these three fairness metrics. On the other hand, optimizing for fairness notions that are not procedural fairness leads to significant drops in procedural fairness, indicating that it is difficult to incidentally satisfy without explicitly optimizing for it. Figure \ref{fig:all_distributions} illustrates this point. Another interesting finding is that while our selected procedurally fair policies do not perfectly satisfy the other two fairness notions, they achieve high fairness scores with low standard deviation. On the other hand, algorithms other than procedural fairness perform poorly on the procedural fairness metric and have significantly larger standard deviations, indicating that optimizing for equality or utility maximization does not inadvertently optimize for equal voice and is quite unstable in outcomes. Table \ref{tab:performance_metrics} reports the exact numerical results.

When it comes to preference distributions, only the procedural fairness score sees meaningful change, particularly on Mallows. At small values of $\phi$, our selected procedurally fair policies yield almost identical policies to utilitarianism and Nash welfare. However, as disagreement among agents grows, every outcome policy's procedural fairness erodes (except for inequality minimization), indicating that the tradeoff between representation and outcomes scales with disagreement. Inequality minimization is the only policy whose procedural score grows with disagreement, uncovering an interesting dynamic: under consensus, equalizing outcomes requires overriding the shared favourite, but under discord, equal outcomes and equal voice partially align. On the other hand, the high achievement of procedural fairness on the other two axes stays consistent.

\subsection{Real-World Example}

To demonstrate our approach on real-world preferences, we instantiate a bandit using the Mechanical Turk Dots dataset from PrefLib \cite{Preflib, Mao_Procaccia_Chen_2013}. This dataset contains complete preference orderings over four candidates ($K=4$), offering a tractable and interpretable setting. We uniformly sample 50 of the 800 voters and map an agent's first through fourth preferences to mean rewards of 0.9, 0.63, 0.37, and 0.1, respectively. We then evaluate our procedural fairness learning algorithm on this instance. Rewards are drawn from Beta distributions with agent-specific means and a fixed standard deviation of 0.1. Setting $\gamma = 0.7$ and $\alpha = 1$ (which were selected to have reasonably limited regret based on the regret bound), we conduct a single $100,000$-step run.

\begin{table}[htbp]
\centering
\begin{tabular}{lrrr}
\toprule
 & PF Score & IM Score & UT Score \\
\midrule
PF Policy  & $1.00$ & $0.99$ & $0.90$ \\
\midrule
IM Policy  & $0.74$ & $1.00$ & $0.80$ \\
UT Policy  & $0.38$ & $0.86$ & $1.00$ \\
EG Policy  & $0.74$ & $1.00$ & $0.80$ \\
NSW Policy & $0.76$ & $0.95$ & $0.99$ \\
GGI Policy & $0.74$ & $1.00$ & $0.80$ \\
\bottomrule
\end{tabular}
\caption{Performance metrics for each policy. Rows denote the policies optimizing each objective; columns denote fairness and utility scores.}
\label{tab:mturk_metrics}
\end{table}

\begin{figure}[t]
\centering
\includegraphics[width=\linewidth]{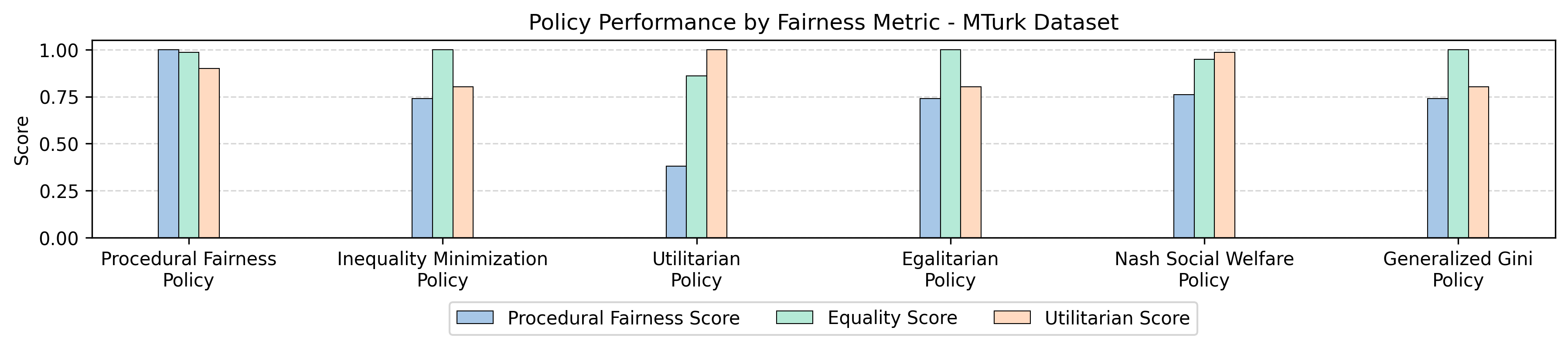}
\caption{Each fairness notion's optimal policy scored on the three metrics. Each bar indicates a fairness metric, and the columns indicate each fairness notion or algorithm's optimal policy. This graph has no error bars because it is the result from a single bandit instance.}
\label{fig:mturk}
\end{figure}

\begin{figure}[t]
\centering
\begin{subfigure}[t]{0.48\linewidth}
    \centering
    \includegraphics[width=\linewidth]{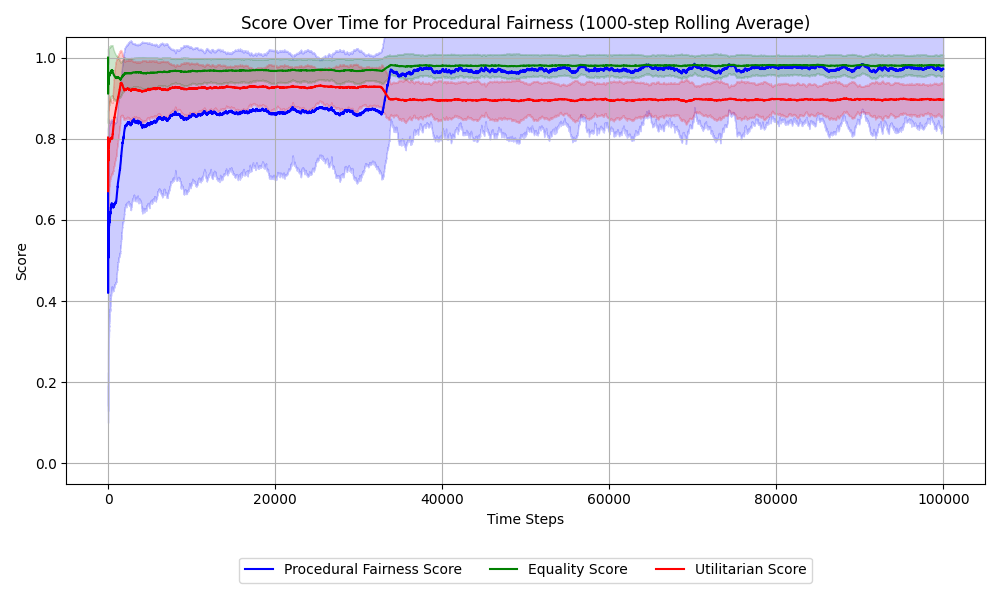}
    \caption{Procedural fairness algorithm's scores over time, shown as 1,000-step rolling averages with $\pm 1$ standard deviation.}
    \label{fig:learning_curve}
\end{subfigure}
\hfill
\begin{subfigure}[t]{0.48\linewidth}
    \centering
    \includegraphics[width=\linewidth]{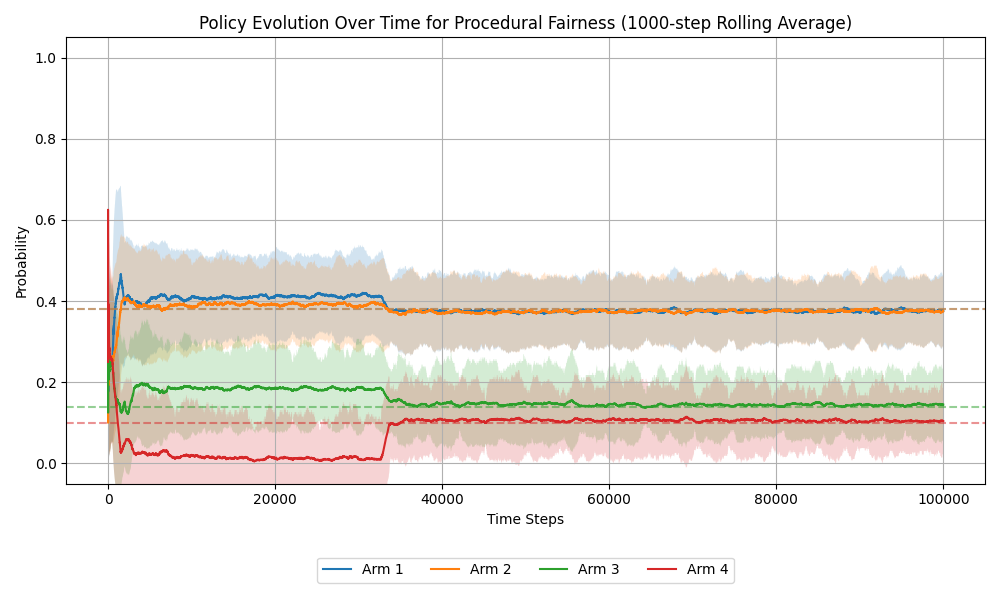}
    \caption{The procedural fairness algorithm's policy over time. The shading indicates 1 standard deviation with a rolling average of 1,000 steps.}
    \label{fig:policy_evol}
\end{subfigure}
\caption{Comparison of the procedural fairness algorithm's performance over time on the MTurk dataset.}
\label{fig:combined}
\end{figure}

Table \ref{tab:mturk_metrics} compares the optimal policies of each benchmark scored across our three fairness metrics, while Figure \ref{fig:mturk} visualizes the same comparison. Figure \ref{fig:combined} illustrates the policy's convergence and the evolution of the fairness metrics over time. Procedural fairness achieves the best overall balance across the three metrics, maintaining high equality and utility while guaranteeing equal representation, and is the only objective that achieves a score of 0.9 or more on every metric. Meanwhile, the utilitarian policy maximizes utility but performs poorly on procedural fairness. Overall, the results from this run are consistent with our findings from the synthetic data.

%% file: sections/discussion-conclusions.tex
\section{Discussion \& Conclusions}

\subsection{Fairness as Legitimacy}

As established, fairness in multi-agent systems has too often been reduced to outcomes over procedure. It is easy to understand why: these are elegant and simple to compute measures. Nonetheless, these are values imposed from the outside. When fairness is defined by an external hand, it risks being viewed as illegitimate.

Procedural fairness offers us another path: a path of respecting the dignity of equal participation and influence in the decision-making process. It does not ask which balance of outcomes is most ideal, but how the decisions are made and who influenced them. This principle is not novel. It is the same principle that gave us the Magna Carta, that sustains constitutions and democracies centuries later, and is backed up by evidence as being preferable \cite{ANAND2001247, lind2013social, tyler1990people}. Such systems endure not because they guarantee optimal outcomes, but because the participants recognize the legitimacy of the decisions made.

For multi-agent decision-making, this means changing the design of systems from asking ``what distribution of rewards is best?'' to ``whose preferences shaped the decision?'' When we focus on centering the voice of the agents and ensuring fair process, we move away from normative judgements that impose values from the outside, towards systems that agents regard as legitimate and more accurately reflect what humans view as preferable.

\subsection{Tradeoffs as Design Choices}

What it means to be fair has never been a unifying principle. It is a contest of rival claims, philosophies, and moral visions, each irreconcilable with the others. This paper makes this blunt claim: no system can guarantee to satisfy all fairness criteria at once, and every fairness choice declares a normative judgment. To choose one is often to forsake another.

This is not a defect in our framework, but a reflection of the human condition. Fairness is not discovered in equations, but declared in values. Making decisions on behalf of a group therefore requires us to choose: whose values will govern? Whose interest will reign supreme? The designer's, or the agents' expressed through equal voice? Historically, we have relied on external metrics to account for fairness: efficiency, equality, or some balance of these two. But elegance and simplicity are not legitimacy. Procedural fairness does not erase tradeoffs; it exposes them, and insists that no agent is denied representation in the decisions that shape their fate.

\subsection{Robustness of Procedural Fairness}

Procedural fairness on its own is not only normatively appealing; our proposed fairness objective is also game-theoretically robust. Our procedurally fair objective lies in the procedural core: no coalition can use its proportional share of decision-making power to preserve every member's representation while increasing at least one member's voice. A coalition may still prefer a compromise that improves expected rewards, but this does not contradict the result: such a compromise trades direct representation for a better outcome. The procedural core asks whether a group can obtain a stronger collective voice, not whether the outcome is the only thing its members value. This reflects an important property of our contribution: representation is valuable in its own right, not merely as a means to better outcomes.

Our experiments tell a complementary story. Our procedurally fair policies preserve equal voice by design, while maintaining strong efficiency and low inequality across settings. Together, these results make procedural fairness more than a normative principle. It is a design principle that yields legitimacy, stability, and resilience. In a field where agents must not only cooperate but endure, procedural fairness stands as a strong baseline for multi-agent systems.

\subsection{Broader Applications of Procedural Fairness}

One of the most important points of this paper is that procedural fairness is not confined to the abstract setting of multi-agent bandits. It appeals to a greater idea that legitimacy turns less on the outcome achieved than on the process by which it was reached.

Such a takeaway must not be forgotten for artificial systems. Whether it be allocating computational resources, governing platforms, or coordinating autonomous agents, the critical question is not \emph{what} was decided, but \emph{how}. To embed procedural fairness into multi-agent systems is therefore not to borrow a human tradition for symbolic value, but to ensure that the systems we build today and tomorrow reflect not only intelligence, but humanity's enduring commitment to fairness through equal voice.